\documentclass[11pt]{article}

\usepackage{amsmath}
\usepackage{amsfonts}
\usepackage{amsthm}
\usepackage{amssymb}
\usepackage{latexsym}   
\usepackage{graphicx}
\usepackage{color}
\usepackage{hyperref}
\usepackage[linesnumbered,ruled,vlined]{algorithm2e}
\usepackage{booktabs}
\usepackage{authblk}
\usepackage{fullpage}
\usepackage{subcaption}
\usepackage{cite}

\newcommand{\beql}[1]{\begin{equation}\label{#1}}
\newcommand{\eeq}{\end{equation}}

\newtheorem{theorem}{Theorem}
\newtheorem{corollary}{Corollary}
\newtheorem{problem}{Problem}
\theoremstyle{definition}

\newcommand{\hide}[1]{}

\newcommand{\argmax}{\mathop{\rm argmax}}
\newcommand{\argmin}{\mathop{\rm argmin}}
\newcommand{\OPT}{\textsf{OPT}\xspace}
\newcommand{\wmax}{w_{\rm max}}
\newcommand{\wmin}{w_{\rm min}}

\title{Finding Densest $k$-Connected Subgraphs} 
\author[1]{Francesco Bonchi\thanks{francesco.bonchi@isi.it}}
\author[1]{David Garc\'ia-Soriano\thanks{d.garcia.soriano@isi.it}}
\author[2]{Atsushi Miyauchi\thanks{miyauchi@mist.i.u-tokyo.ac.jp}}
\author[1,3]{Charalampos~E.~Tsourakakis\thanks{tsourolampis@gmail.com}}

\affil[1]{ISI Foundation, Turin, Italy}
\affil[2]{University of Tokyo, Tokyo, Japan}
\affil[3]{Boston University, Boston, USA}

\begin{document}
\maketitle

\begin{abstract}
Dense subgraph discovery is an important graph-mining primitive with a variety of real-world applications. 
One of the most well-studied optimization problems for dense subgraph discovery is the densest subgraph problem, 
where given an edge-weighted undirected graph $G=(V,E,w)$, 
we are asked to find $S\subseteq V$ that maximizes the density $d(S)$, i.e., half the weighted average degree of the induced subgraph $G[S]$. 
This problem can be solved exactly in polynomial time and well-approximately in almost linear time. 
However, a densest subgraph has a structural drawback, 
    namely, the subgraph may not be robust to vertex/edge failure. 
Indeed, a densest subgraph may not be well-connected, 
which implies that the subgraph may be disconnected by removing only a few vertices/edges within it. 
In this paper, we provide an algorithmic framework to find a dense subgraph that is well-connected in terms of vertex/edge connectivity. 
Specifically, we introduce the following problems: given a graph $G=(V,E,w)$ and a positive integer/real $k$, 
we are asked to find $S\subseteq V$ that maximizes the density $d(S)$ under the constraint that $G[S]$ is $k$-vertex/edge-connected. 
For both problems, we propose polynomial-time (bicriteria and ordinary) approximation algorithms, 
using classic Mader's theorem in graph theory and its extensions. 

%\color{red}{
%\begin{itemize}
%\item Should we emphasize the constant-factor approximation ratio (independent of edge weights) for the case where the if-condition of Algorithm 1 is true? I think this is a significant advantage of the use of our weighted Mader. 
%\item NP-hardness proof.
%\end{itemize}
%}
\end{abstract}

\clearpage

\section{Introduction} 
\label{sec:intro} 
Dense subgraph discovery is an important graph-mining primitive with a variety of real-world applications~\cite{Gionis_Tsourakakis_15}. 
Examples include detecting communities and spam link farms in the Web graph~\cite{Dourisboure+07,Gibson+05}, 
extracting molecular complexes in protein--protein interaction networks~\cite{Bader_Hogue_03,Spirin03}, 
finding experts in crowdsourcing systems~\cite{Kawase+19}, 
and real-time story identification from tweets~\cite{Angel+12}. 

One of the most well-studied optimization problems for dense subgraph discovery is the \emph{densest subgraph problem}. 
Let $G=(V,E,w)$ be a simple undirected graph with edge weight $w:E\rightarrow \mathbb{R}_{>0}$, 
%Throughout this paper, we assume that $w(e)>0$ for any $e\in E$. 
%Throughout this paper, let $G=(V,E,w)$ be an undirected graph with edge weight $w:E\rightarrow \mathbb{R}_{>0}$, 
where $\mathbb{R}_{>0}$ is the set of positive reals. 
Throughout this paper, we assume that $|E|\geq 1$, edge-weighted graphs have only positive weights, and $G$ is connected. 
For $S\subseteq V$, let $G[S]$ denote the subgraph induced by $S$, 
i.e., $G[S]=(S,E(S))$, where $E(S)=\{\{u,v\}\in E\mid u,v\in S\}$. 
The \emph{density} of $S\subseteq V$ is defined as $d(S)=w(S)/|S|$, 
where $w(S)$ is the sum of edge weights of $G[S]$, i.e., $w(S)=\sum_{e\in E(S)}w(e)$. 
In the densest subgraph problem, given a graph $G=(V,E,w)$, we are asked to find $S\subseteq V$ that maximizes $d(S)$. 
An optimal solution to this problem is referred to as a \emph{densest subgraph}. 

Unlike most optimization problems for dense subgraph discovery such as the maximum clique problem~\cite{Garey_Johnson_79}, 
the densest subgraph problem is polynomial-time solvable. 
Indeed, there are some polynomial-time exact algorithms 
such as Goldberg's flow-based algorithm~\cite{Goldberg84} and Charikar's linear-programming-based algorithm~\cite{Charikar00}. 
Moreover, it was shown by Charikar~\cite{Charikar00} that a simple greedy algorithm admits $1/2$-approximation in almost linear time. 
Partially due to its solvability, the densest subgraph problem has been employed in many real-world applications. 

However, it can be seen that a densest subgraph has a structural drawback, 
that is, the subgraph may not be robust to vertex/edge failure. 
To see this, let us introduce some terminology. 
A vertex subset $S\subset V$ is called a \emph{vertex separator} of $G$ 
if its removal disconnects $G$, i.e., partitions $G$ into at least two non-empty graphs between which there are no edges.  
Note that no clique has a vertex separator. 
An edge subset $F\subseteq E$ is called a \emph{cut} of $G$ 
if its removal disconnects $G$.
%if the same condition as above is met. 
The \emph{weight} of a cut is defined to be the sum of weights of edges within it. 
The \emph{vertex connectivity} of $G$, denoted by $\kappa(G)$, 
is the smallest cardinality of a vertex separator of $G$ if $G$ is not a clique and $|V|-1$ otherwise. 
The \emph{edge connectivity} of $G$, denoted by $\lambda(G)$, is the smallest weight of a cut of $G$. 

A densest subgraph does not necessarily have large vertex/edge connectivity, which means that the subgraph may be disconnected by removing only a few vertices/edges within it. 
For instance, consider an unweighted graph $G$ (i.e., $w(e)=1$ for every $e\in E$) 
consisting of two equally-sized large cliques that share only a few vertices or are connected by only a few edges. 
In both cases, the entire graph is a densest subgraph, 
but it is easily disconnected by removing the common vertices in the former case and the bridging edges in the latter case. 

%However, it can be seen that a densest subgraph has a structural drawback, 
%that is, the subgraph may not be robust for vertex/edge failure. 
%Indeed, a densest subgraph is not necessarily well-connected in terms of vertex/edge connectivity, 
%which implies that the subgraph may be decomposed by removing only a few vertices/edges within it. 
%For instance, consider an unweighted graph $G$ (i.e., $w(e)=1$ for every $e\in E$) 
%consisting of two equally-sized large cliques that share only a few vertices or are connected by only a few edges. 
%In both cases, the entire graph is a densest subgraph, 
%but it is easily decomposed by removing the common vertices in the former case and by removing the bridging edges in the latter case. 

\begin{figure}[t]
\centering
\begin{minipage}{0.49\textwidth}
\centering
\includegraphics[scale=0.42]{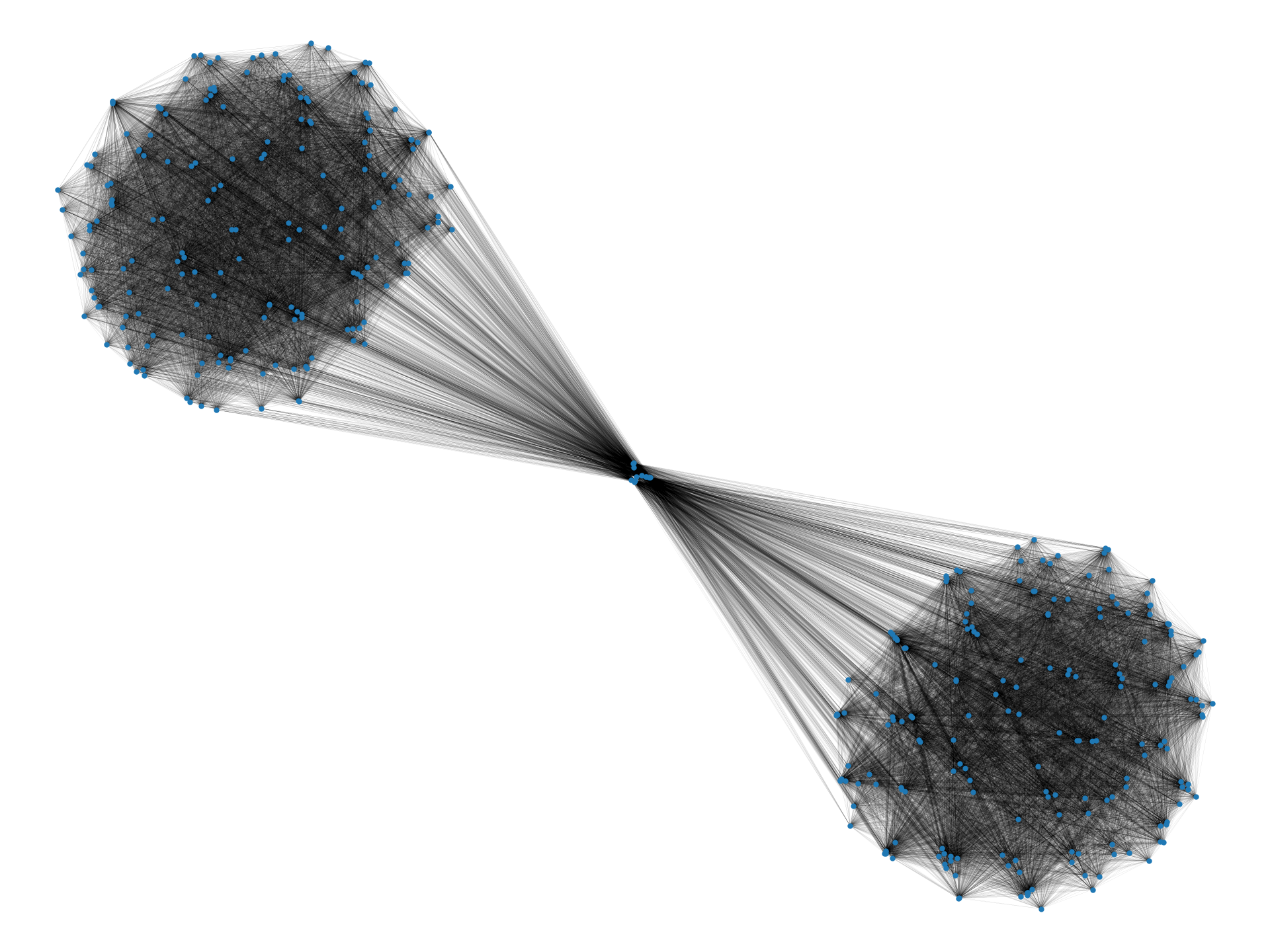}
\subcaption{\texttt{web-BerkStan}}
\end{minipage}
\begin{minipage}{0.49\textwidth}
\centering
\includegraphics[scale=0.42]{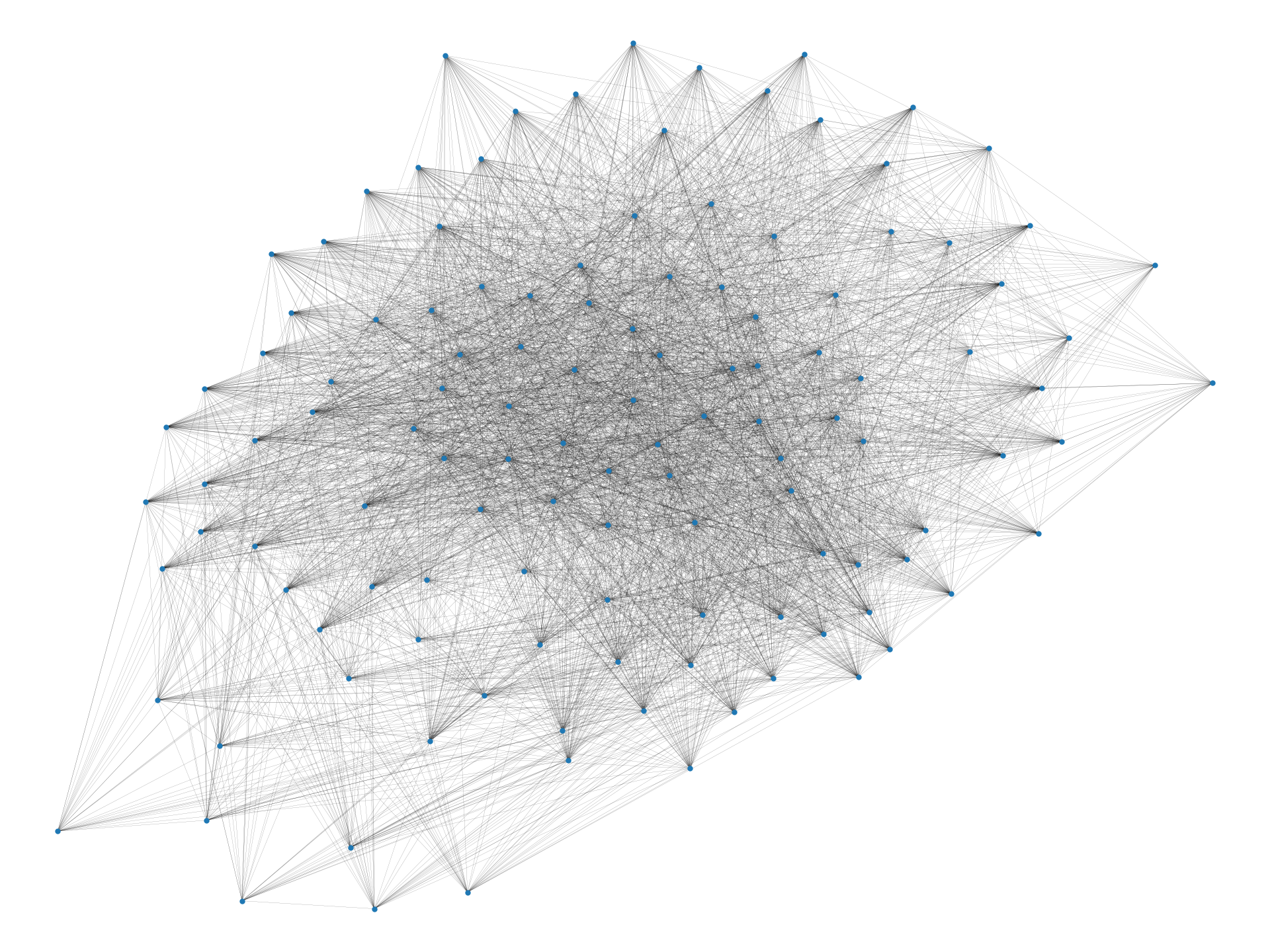}
\subcaption{\texttt{web-Google}}
\end{minipage}\\
\begin{minipage}{0.49\textwidth}
\centering
\includegraphics[scale=0.42]{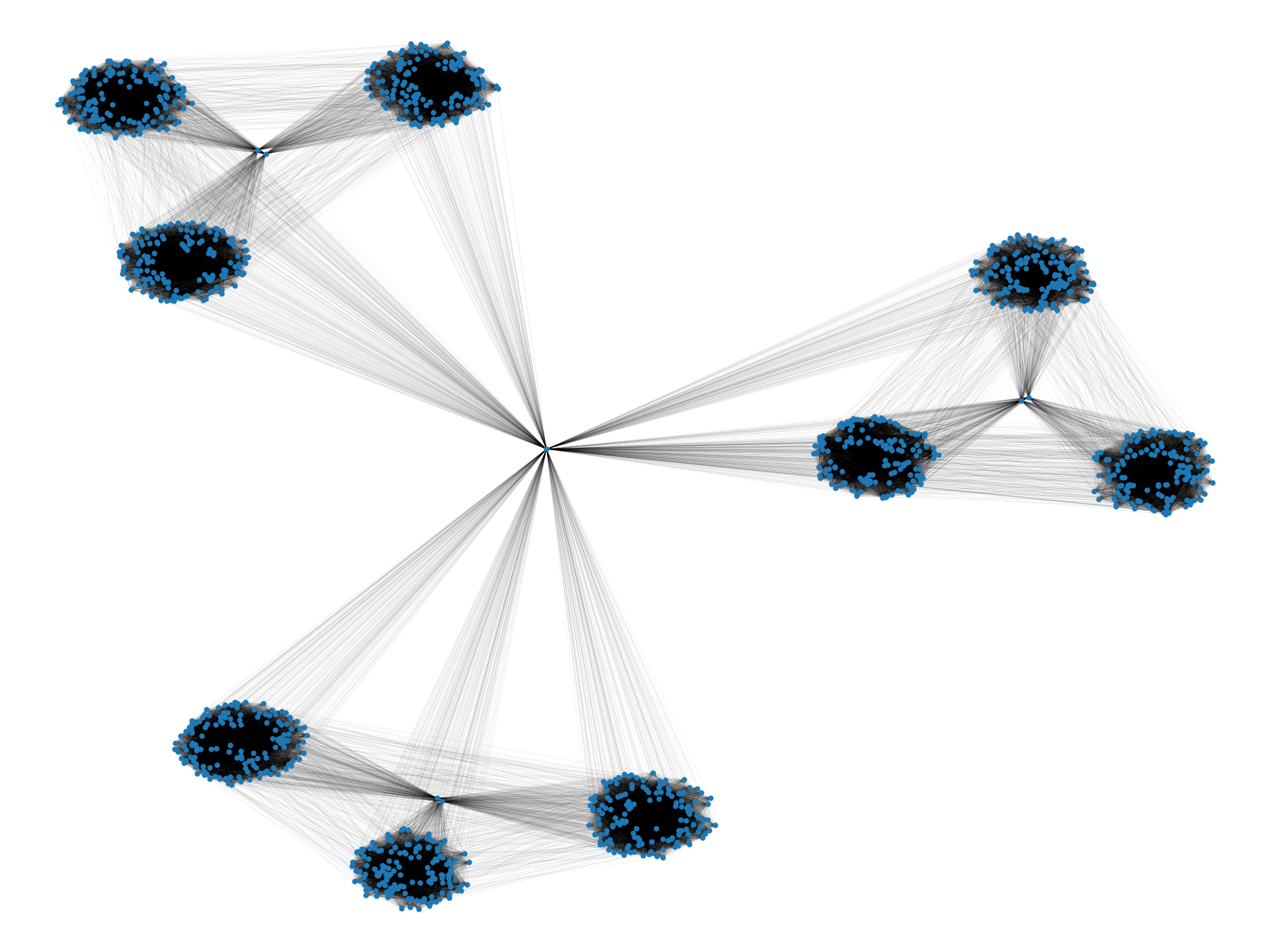}
\subcaption{\texttt{web-NotreDame}}
\end{minipage}
\begin{minipage}{0.49\textwidth}
\centering
\includegraphics[scale=0.42]{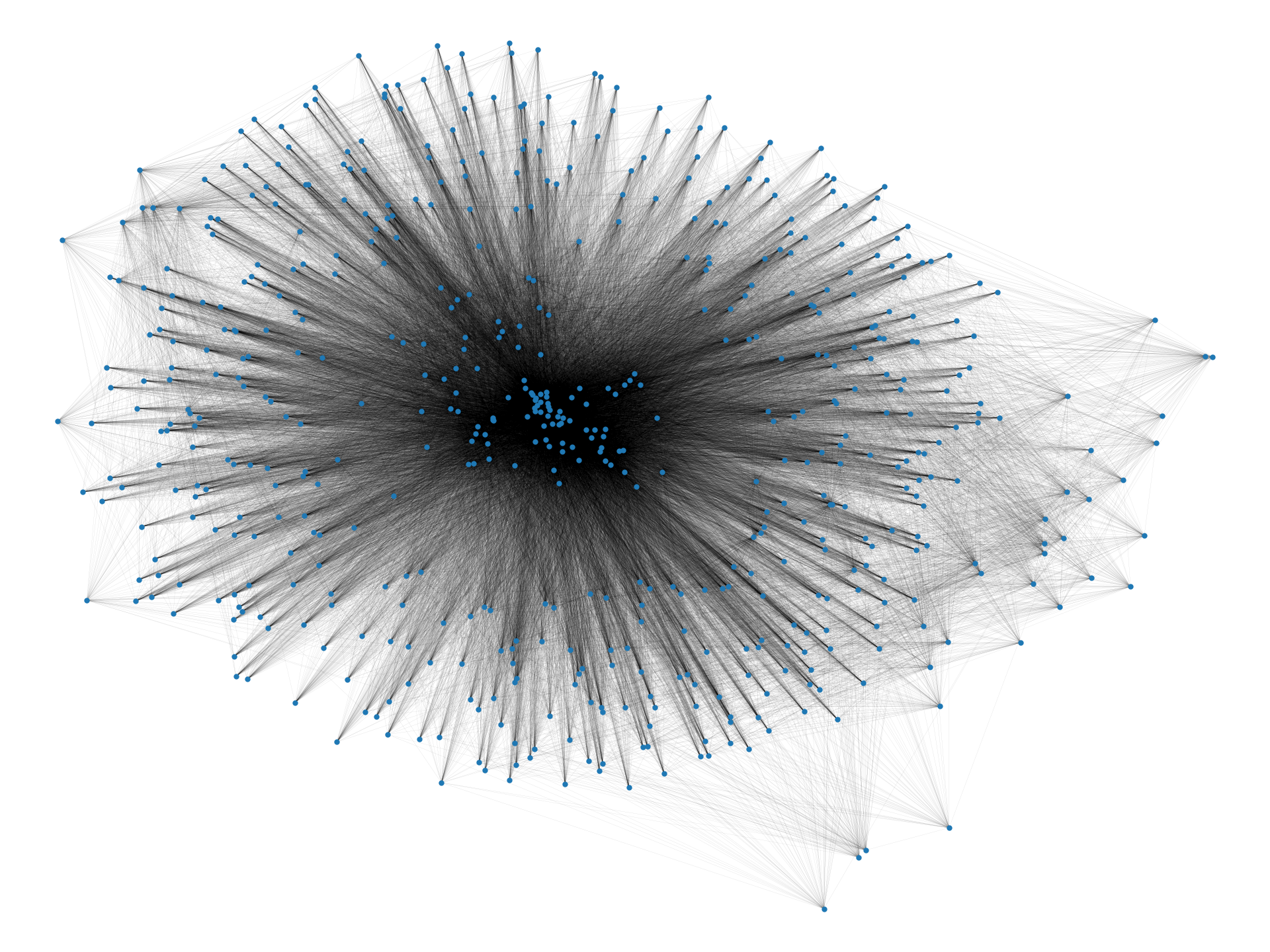}
\subcaption{\texttt{web-Stanford}}
\end{minipage}
\caption{Densest subgraphs in real-world Web graphs.}\label{fig:densest}
\end{figure}

In this paper, we provide an algorithmic framework to find a dense subgraph that is well-connected in terms of vertex/edge connectivity. 
An (edge-weighted) graph $G$ is said to be \emph{$k$-vertex-connected} 
if $\kappa(G)$ is no less than $k$. 
%if the number of vertices is greater than $k$ and the removal of any $k-1$ vertices leaves $G$ connected. 
On the other hand, an edge-weighted graph $G$ is said to be \emph{$k$-edge-connected}
if $\lambda(G)$ is no less than $k$. 
Using these criteria, we introduce the following problems:
\begin{problem}[Densest $k$-vertex-connected subgraph]
\label{prob:vertex}
Given an edge-weighted undirected graph $G=(V,E,w)$, where $w:E\rightarrow \mathbb{R}_{>0}$, 
and a positive integer $k\in \mathbb{Z}_{>0}$,
the goal is to find $S\subseteq V$ that maximizes the density $d(S)$ subject to the constraint that the induced subgraph $G[S]$ is $k$-vertex-connected. 
\end{problem}

\begin{problem}[Densest $k$-edge-connected subgraph]
\label{prob:edge}
Given an edge-weighted undirected graph $G=(V,E,w)$, where $w:E\rightarrow \mathbb{R}_{>0}$, 
and a positive real $k\in \mathbb{R}_{>0}$,
the goal is to find $S\subseteq V$ that maximizes the density $d(S)$ subject to the constraint that the induced subgraph $G[S]$ is $k$-edge-connected. 
\end{problem}

In the two-cliques example we discussed earlier, an optimal solution to Problems~\ref{prob:vertex} and~\ref{prob:edge} with a  sufficiently large value for $k$ would be one of the cliques, which is robust to vertex/edge failure and nearly as dense as the densest subgraph (i.e., the entire graph).  We observe that Problems~\ref{prob:vertex} and~\ref{prob:edge} are meaningful for real-world data too; Figure~\ref{fig:densest} visualizes densest subgraphs of the four real-world Web graphs that are publicly available at SNAP (Stanford Network Analysis Project)~\cite{snap} using a spring layout positioning.\footnote{Graphs have been made simple undirected by ignoring the direction of edges, and by removing self-loops and multiple edges.}
As we can visually observe,    small separators may exist in real-world densest subgraphs.  Table~\ref{tab:densest} summarizes the  detailed statistics. 
As can be seen, the densest subgraphs in \texttt{web-BerkStan} and \texttt{web-NotreDame} have surprisingly small vertex connectivity; 
in fact, they have vertex connectivity of twelve and one, respectively. Note that for both densest subgraphs, vertex connectivity is much smaller than the minimum degree of vertices, a trivial upper bound on that.

\begin{table}[t]
\centering
\caption{Details of densest subgraphs $G[S^\text{DS}]$ in four real-world Web graphs: $\delta(G[S^\text{DS}])$ represents the minimum degree of vertices in $G[S^\text{DS}]$, a trivial upper bound on $\kappa(G[S^\text{DS}])$ and $\lambda(G[S^\text{DS}])$.}\label{tab:densest}
\begin{tabular}{lrrrrrr}
\toprule
Graph  &$|S^\text{DS}|$  &$|E(S^\text{DS})|$ &$d(S)$ &$\kappa(G[S^\text{DS}])$ &$\lambda(G[S^\text{DS}])$ &$\delta(G[S^\text{DS}])$\\
\midrule
\texttt{web-BerkStan} &392  &40,535  &103.41      &12 &201 &201\\
\texttt{web-Google}   &123  &3,449   &28.04      &30  &30 &30\\
\texttt{web-NotreDame}&1,367&107,526 &78.66       &1  &155 &155\\
\texttt{web-Stanford} &597  &35,456  &59.39       &60 &60  &60 \\
\bottomrule
\end{tabular}
\end{table}

For both problems, we propose polynomial-time (bicriteria and ordinary) approximation algorithms. 
Let $\wmax$ and $\wmin$ denote the maximum and minimum weights, respectively, over all edges in $G$,  
i.e., $\wmax = \max_{e\in E}w(e)$ and $\wmin = \min_{e\in E}w(e)$. 

Our first result is polynomial-time $\left(\frac{\gamma}{4}\cdot \frac{\wmin}{\wmax}, 1/\gamma\right)$-bicriteria approximation algorithms with parameter $\gamma \in [1,2]$ for Problems~\ref{prob:vertex} and~\ref{prob:edge}. 
That is, the algorithm for Problem~\ref{prob:vertex}/Problem \ref{prob:edge} outputs $S\subseteq V$ having density at least the optimal value times $\frac{\gamma}{4}\cdot
\frac{\wmin}{\wmax}$ but only satisfies a $(k/\gamma)$-vertex/edge-connectivity constraint (rather than the original $k$-vertex/edge-connectivity constraint). 
Note that if we set $\gamma=1$, we can obtain $\left(\frac{1}{4}\cdot \frac{\wmin}{\wmax}\right)$-approximation algorithms. 
The design of our algorithms is based on an elegant theorem in graph theory, proved by Mader~\cite{Mader72}. 
The theorem states that any (unweighted) dense graph contains a highly vertex-connected subgraph 
wherein the minimum degree of vertices is greater than the density of the entire vertex set. 
We refer to this subgraph as a \emph{Mader subgraph} 
and our algorithm finds a Mader subgraph of a densest subgraph of each maximal $k$-vertex-connected subgraph of $G$. 
It should be noted that to deal with edge-weighted graphs, we generalize Mader's theorem. 
Our generalized version cannot be directly obtained from the original statement of Mader's theorem, 
and is essential to derive the bicriteria approximation ratio for edge-weighted graphs. 
%The design of our algorithms is based on a combination 
%of computing densest subgraphs and computing highly vertex/edge-connected subgraphs with some minimum weighted degree constraint. 
%To verify the bicriteria approximation ratio, we use an elegant theorem in graph theory proved by Mader~\cite{Mader72}, 
%which states that any (unweighted) dense graph contains a highly vertex-connected subgraph 
%wherein the minimum weighted degree of vertices is greater than the density of the graph. 
%To deal with edge-weighted graphs, we generalize this Mader's theorem. 
%We wish to remark that our generalized version cannot be directly obtained from the original statement of Mader's theorem, 
%and ours is essential to derive the bicriteria approximation ratio for edge-weighted graphs. 

Our second result is polynomial-time $\left(\frac{6}{19}\cdot \frac{\wmin}{\wmax}\right)$-approximation algorithms for Problems~\ref{prob:vertex} and~\ref{prob:edge}, 
which improves the above approximation ratio of $\frac{1}{4}\cdot \frac{\wmin}{\wmax}$ 
derived directly from the bicriteria approximation ratio. 
Our algorithm for Problem~\ref{prob:vertex}/Problem~\ref{prob:edge} computes the most highly connected subgraph in terms of vertex/edge connectivity, which can be done using the algorithms in Matula~\cite{Matula78}. 
In the analysis of the approximation ratio, 
we use a useful variant of Mader's theorem, recently proved by Bernshteyn and Kostochka~\cite{Bernshteyn_Kostochka_16}.

\paragraph{Paper organization.} The remainder of this paper is organized as follows. In Section~\ref{sec:related}, we review related work.  In Section~\ref{sec:mader}, we extend Mader's theorem to edge-weighted graphs and design an algorithm for finding a Mader subgraph. 
In Sections~\ref{sec:biapprox} and~\ref{sec:unweighted}, we present our bicriteria and ordinary approximation algorithms, respectively. We conclude with some open problems in Section~\ref{sec:concl}.

\section{Related Work}
\label{sec:related}
\paragraph{Variations of the densest subgraph problem.}  Wu et al.~\cite{WuZLFJZ16}  consider the problem of detecting a dense and connected subgraph in \emph{dual networks}. 
A dual network is a pair of graphs $G=(V,E_G)$ and $H=(V,E_H)$ defined on the same vertex set $V$, 
which encode different types of connections using two edge sets $E_G$ and $E_H$. 
Wu et al.~\cite{WuZLFJZ16} introduced the following problem: 
%Wu et al.\cite{Wu+15,WuZLFJZ16} introduced the following problem: 
given a dual network $(G,H)$, we are asked to find $S\subseteq V$ that maximizes $d(S)$ in $G$ under the constraint that $H[S]$ is connected (i.e., $1$-edge-connected). 
They proved that the problem is NP-hard and devised a scalable heuristic. 
Problem~\ref{prob:edge} with $k=1$, i.e., the densest $1$-edge-connected subgraph, on unweighted graphs,  
can be seen as a special case of their problem wherein two graphs $G$ and $H$ are the same, i.e., $E_G=E_H$. 
It is easy to see that unlike the general form of their problem, 
the densest $1$-edge-connected subgraph problem (on unweighted graphs) is polynomial-time solvable. 

Two closely related papers are due to Tsourakakis~\cite{Tsourakakis_15} and Kawase and Miyauchi~\cite{Kawase_Miyauchi_18}. They aim to find a near-clique (which is robust to vertex/edge failure) by extending the densest subgraph problem. 
Tsourakakis~\cite{Tsourakakis_15} introduced the problem called the \emph{$k$-clique densest subgraph problem}. 
In this problem, given an unweighted graph $G=(V,E)$, we are asked to find $S\subseteq V$ that maximizes the \emph{$k$-clique density} $w_k(S)/|S|$, 
where $w_k(S)$ is the number of $k$-cliques (i.e., cliques with size $k$) in $G[S]$. 
Tsourakakis~\cite{Tsourakakis_15} showed that this problem (with constant $k$) remains polynomial-time solvable, 
and later, Mitzenmacher et al.~\cite{Mitzenmacher+15} proposed a scalable algorithm that obtains a nearly-optimal solution.  
On the other hand, Kawase and Miyauchi~\cite{Kawase_Miyauchi_18} introduced the problem called the \emph{$f$-densest subgraph problem with convex $f$}. 
In this problem, given an edge-weighted graph $G=(V,E,w)$, we are asked to find $S\subseteq V$ that maximizes 
$w(S)/f(|S|)$, where $f:\mathbb{Z}_{\geq 0}\rightarrow \mathbb{R}_{\geq 0}$ is a monotonically non-decreasing function 
that satisfies $(f(x+2)-f(x+1))-(f(x+1)-f(x))\geq 0$ for any $x\in \mathbb{Z}_{\geq 0}$.  
%Note that $\mathbb{Z}_{\geq 0}$ denotes the set of nonnegative integers. 
This formulation generalizes the NP-hard optimal quasi-cliques problem due to Tsourakakis et al. \cite{tsourakakis2013denser,tsourakakis2015streaming}. 
%This problem is also a generalization of the densest subgraph problem, and can penalize large solution 
Kawase and Miyauchi~\cite{Kawase_Miyauchi_18} studied the hardness of the problem, 
and proposed a polynomial-time approximation algorithm. Although the above two problems contribute to computing a dense subgraph that is robust to vertex/edge failure, they cannot explicitly impose $k$-vertex/edge connectivity. 

There are also some variants that take into account the robustness to the uncertainty of input graphs.  Zou~\cite{Zou_13} studied the densest subgraph problem on \emph{uncertain graphs}. 
Uncertain graphs are a generalization of graphs, which can model the uncertainty of the existence of edges. 
More formally, an uncertain graph consists of an unweighted graph $G = (V,E)$ and a function $p: E \rightarrow [0, 1]$, 
where $e \in E$ is present with probability $p(e)$ whereas $e \in E$ is absent with probability $1-p(e)$. 
In the problem introduced by Zou~\cite{Zou_13}, given an uncertain graph $G=(V,E)$ with $p$, 
we are asked to find $S\subseteq V$ that maximizes the expected value of the density. 
Zou~\cite{Zou_13} observed that this problem can be reduced to the original densest subgraph problem, 
and designed polynomial-time exact algorithm using the reduction. 
Very recently, Tsourakakis et al.~\cite{Tsourakakis+19} introduced the problem called the \emph{risk-averse DSD}. 
In this problem, given an uncertain graph $G=(V,E)$ with $p$, 
we are asked to find $S\subseteq V$ that has a large expected density and at the same time has a small \emph{risk}. 
The risk of $S\subseteq  V$ is measured by the probability that $S$ is not dense on a given uncertain graph. 
They showed that the risk-averse DSD can be reduced to the densest subgraph problem with \emph{negative} edge weights (which is NP-hard), 
and designed an efficient approximation algorithm based on the reduction.

Miyauchi and Takeda~\cite{Miyauchi_Takeda_18} considered the uncertainty of edge weights rather than the existence of edges. 
To model that, they assumed that they have an \emph{edge-weight space} $W=\times_{e\in E}[l_e,r_e]\subseteq \times_{e\in E}[0,\infty)$ that contains the unknown true edge weight $w$. 
To evaluate the performance of $S\subseteq V$ without any concrete edge weight, 
they employed a well-known measure in the field of robust optimization, called the \emph{robust ratio}. 
In their scenario, the robust ratio of $S\subseteq V$ under $W$ is defined as the multiplicative gap between the density of $S$ in terms of edge weight $w'$ and the density of $S^*_{w'}$ in terms of edge weight $w'$ under the worst-case edge weight $w'\in W$, 
where $S^*_{w'}$ is a densest subgraph of $G$ with $w'$. 
Intuitively, $S\subseteq V$ with a large robust ratio has a density close to the optimal value even on $G$ with the edge weight selected adversarially from $W$. 
Using the robust ratio, they formulated the \emph{robust densest subgraph problem}, 
where given an unweighted graph $G=(V,E)$ and an edge-weight space $W=\times_{e\in E}[l_e,r_e]\subseteq \times_{e\in E}[0,\infty)$, 
we are asked to find $S\subseteq V$ that maximizes the robust ratio under $W$. 
Miyauchi and Takeda~\cite{Miyauchi_Takeda_18} designed an algorithm that returns $S\subseteq V$ with a robust ratio of at least $\frac{1}{\max_{e\in E}\frac{r_e}{l_e}}$ under some mild condition. 

In addition to the above, there are many other problem variations. 
The most well-studied variants are size restricted ones~\cite{Andersen_Chellapilla_09,Bhaskara+10,Feige+01,Khuller_Saha_09}.
For example, in the \emph{densest $k$-subgraph problem}~\cite{Feige+01}, 
given an edge-weighted graph $G=(V,E,w)$ and a positive integer $k\in \mathbb{Z}_{>0}$,
we are asked to find $S\subseteq V$ that maximizes $d(S)$ subject to the constraint $|S|=k$.
It is known that such a restriction makes the problem much harder;
indeed, the densest $k$-subgraph problem is NP-hard and the best known approximation ratio is $\Omega(1/n^{{1/4}+\epsilon})$ for any $\epsilon>0$~\cite{Bhaskara+10}.
The densest subgraph problem has also been extended to more general computation models and graph structures.
As for computation models, to cope with the dynamics of real-world graphs,
some literature has considered dynamic settings~\cite{Epasto+15,Hu+17},
and moreover, to model the limited computation resources in reality,
some literature has considered streaming settings~\cite{Angel+12,Bahmani+12,Bhattacharya+15}.
As for graph structures, the problem has been defined on hypergraphs~\cite{Hu+17,Miyauchi+15} and multilayer networks~\cite{Galimberti+17}.
%The density itself has also been generalized for some specific purposes~\cite{Miyauchi_Kakimura_18}. 

\paragraph{Vertex and edge connectivity.}  In the \emph{vertex connectivity problem}, we are asked to compute $\kappa(G)$ for a given graph $G=(V,E)$. 
For this problem, Gabow~\cite{Gabow06} developed an $O(|V|(\kappa(G)^2\cdot \min\{|V|^{3/4}, \kappa(G)^{3/2}\}+\kappa(G) |V|))$-time algorithm, 
which also computes a corresponding \emph{minimum vertex separator} $S\subset V$. 
This is one of the current fastest deterministic algorithms for the problem, 
although there are various randomized algorithms~(e.g., see \cite{Forster+20,Henzinger+00,Linial+88,Nanongkai+19}). 
Note that there are linear-time algorithms that decide whether $G$ is 2-vertex-connected and 3-vertex-connected, respectively, 
due to Tarjan~\cite{Tarjan72} and Hopcroft and Tarjan~\cite{Hopcroft_Tarjan_73}. 
%Even~\cite{Even75} developed an $O(k|E||V|^2)$-time algorithm for deciding whether $G$ is $k$-vertex-connected. 

Another important problem related to vertex connectivity is to compute the family of maximal $k$-vertex-connected subgraphs, 
which will be solved in our bicriteria approximation algorithm for Problem~\ref{prob:vertex}. 
For $S\subseteq V$ and $k\in \mathbb{Z}_{>0}$, 
the induced subgraph $G[S]$ is called a \emph{maximal $k$-vertex-connected subgraph}
if $G[S]$ is $k$-vertex-connected and no superset of $S$ has this property. 
For this task, the first polynomial-time algorithm is given by Matula~\cite{Matula77}. 
Note that maximal $k$-vertex-connected subgraphs may overlap each other; 
the design of the algorithm by Matula~\cite{Matula77} is based on the fact 
that the maximum total number of maximal $k$-vertex-connected subgraphs is $O(|V|)$~\cite{Matula77}. 
Later, Makino~\cite{Makino88} designed an $O(|V|\cdot T)$-time algorithm, 
where $T$ is the computation time required to find a vertex separator of size at most $k-1$. 
Combined with the above vertex connectivity algorithm by Gabow~\cite{Gabow06}, 
the algorithm by Makino~\cite{Makino88} yields the running time of $O(|V|^2(k^2\cdot \min\{|V|^{3/4}, k^{3/2}\}+k |V|))$. 
For some special $k$, there are some existing algorithms that have better running time. 
For $k=2$ and 3, there are linear-time algorithms by Tarjan~\cite{Tarjan72} and Hopcroft and Tarjan~\cite{Hopcroft_Tarjan_73}, respectively. 
For any constant $k$, Henzinger et al.~\cite{Henzinger+15} presented an $O(|V|^3)$-time algorithm.

In the \emph{(global) minimum cut problem},  given an edge-weighted graph $G=(V,E,w)$, we are asked to find the minimum weight cut $F\subseteq E$. 
For this problem, Nagamochi and Ibaraki~\cite{Nagamochi_Ibaraki_92} designed an $O(|V|(|E|+|V|\log|V|))$-time algorithm. 
Later, Stoer and Wagner~\cite{Stoer_Wagner_97} and Frank~\cite{Frank94} independently presented a very simple algorithm 
that still has the same running time. 
For simple unweighted graphs, the seminal work by Karger~\cite{Karger00} provides a randomized (Monte Carlo) algorithm 
that runs in nearly-linear, $O(|E|\log^3|V|)$, time. 
As this algorithm does not necessarily return the right answer, 
Karger~\cite{Karger00} posed an open question to find a nearly-linear-time deterministic algorithm. 
In a recent breakthrough, Kawarabayashi and Thorup~\cite{Kawarabayashi_Thorup_18} answered the question; 
they developed a deterministic algorithm for simple unweighted graphs that runs in $O(|E|\log^{12}|V|)$ time. 
Very recently, Henzinger et al.~\cite{Henzinger20} improved the running time to $O(|E|\log^2|V|\log\log^2|V|)$ time, 
which is better even than that of the randomized algorithm by Karger~\cite{Karger00}. 

As in the vertex connectivity case, computing the family of maximal $k$-edge-connected subgraphs is also an important problem, 
which will be solved in our bicriteria approximation algorithm for Problem~\ref{prob:edge}. 
For $S\subseteq V$ and $k\in \mathbb{R}_{>0}$, 
the induced subgraph $G[S]$ is called a \emph{maximal $k$-edge-connected subgraph}
if $G[S]$ is $k$-edge-connected and no superset of $S$ has this property. 
The problem can be solved using any minimum cut algorithm as follows: 
if the weight of the minimum cut of the graph is less than $k$, 
divide the graph into two subgraphs along with the cut and then repeat the procedure on the resulting subgraphs. 
%recursively, and collect the resulting $k$-edge-connected subgraphs. 
%as long as the weight of the minimum cut of a (sub)graph at hand is less than $k$, 
%divide the graph into two subgraphs along with the cut, recursively, and collect the resulting $k$-edge-connected subgraphs. 
For edge-weighted graphs, we can directly obtain an $O(|V|^2(|E|+|V|\log |V|))$-time algorithm 
using one of the above minimum cut algorithms by Nagamochi and Ibaraki~\cite{Nagamochi_Ibaraki_92}, Stoer and Wagner~\cite{Stoer_Wagner_97}, and Frank~\cite{Frank94}. 
To the best of our knowledge, there is no existing algorithm that has a better running time. 
For simple unweighted graphs, we can again directly obtain an $O(|E||V|\log^2|V|\log\log^2|V|)$-time algorithm 
using the above minimum cut algorithm by Henzinger et al.~\cite{Henzinger20}. 
Unlike the weighted case, for some special $k$, 
there are some existing algorithms that have a better running time. 
For $k=2$, there is a linear-time algorithm by Tarjan~\cite{Tarjan72}. 
For any constant $k$, Henzinger et al.~\cite{Henzinger+15} presented an $O(|V|^2\log|V|)$-time algorithm, 
%For any $k$, Gabow~\cite{Gabow95} $O(k|E||V|\log|V|)$ time. 
%This running time is better than the above $O(|E||V|\log^2|V|\log\log^2|V|)$ time when $k=o(\log|V|\log\log|V|)$. 
and more recently, Chechik et al.~\cite{Chechik17} provided an $O(\sqrt{|V|}(|E|+|V|\log|V|))$-time algorithm. 
The latter algorithm is efficient particularly for sparse graphs; 
indeed, the latter is better than the former when $|E|=o(|V|^{3/2}\log |V|)$. 
%Note that most of the above work deal with directed graphs. 
%For the case of $k=3$, Galil and Italiano~\cite{Galil_Italiano_91}. 
Very recently, for any $k\in \mathbb{Z}_{>0}$, 
Forster et al.~\cite{Forster+20} developed a randomized (Las Vegas) algorithm that has expected running time 
$O(k^3|V|^{3/2}\log|V|+k|E|\log^2|V|)$, which is faster than the algorithm by Chechik et al.~\cite{Chechik17} 
(for general $k\in \mathbb{Z}_{>0}$).

\section{Mader's theorem and Mader subgraph}
\label{sec:mader} 
In this section, we extend Mader's theorem to edge-weighted graphs and design an algorithm for finding a Mader subgraph. 

\subsection{Mader's Theorem on Edge-Weighted Graphs}

Mader's theorem \cite{Mader72} is a foundational theorem in graph theory.  The precise statement is as follows: 

\begin{theorem}[Mader~\cite{Mader72}; see also Theorem~1.4.3 in Diestel~\cite{Diestel16}]\label{thm:mader}
Let $G=(V,E)$ be an unweighted graph and let $d$ be a positive integer. 
If $G$ has density at least $d$, then $G$ has a $(\lfloor d/2\rfloor + 1)$-vertex-connected subgraph 
wherein the minimum degree of vertices is greater than $d$. 
\end{theorem}

A straightforward application of Theorem~\ref{thm:mader} to edge-weighted graphs would yield the following result. Let $G=(V,E,w)$ be an edge-weighted graph, let $d$ be a positive real, and assume that $G$ has density at least $d$. 
Now consider an unweighted graph $G'=(V,E)$ defined on the same vertex set $V$ and edge set $E$. As $G'$ has the density of at least $d/\wmax$ (i.e., at least $\lfloor d/\wmax\rfloor$), 
by Theorem~\ref{thm:mader}, we see that $G'$ has a $\left(\left\lfloor\frac{\lfloor d/\wmax \rfloor}{2}\right\rfloor+1\right)$-vertex-connected subgraph wherein the minimum degree of vertices is greater than $\lfloor d/\wmax\rfloor$. 
Therefore, we can deduce that $G$ has a $\left(\left\lfloor\frac{\lfloor d/\wmax \rfloor}{2}\right\rfloor+1\right)$-vertex-connected subgraph wherein the minimum weighted degree of vertices is greater than $\wmin \lfloor d/\wmax\rfloor$. 
However, this is weaker than what we need to prove the approximation guarantee of our algorithms, as we discuss in Section~\ref{subsec:remark}. 

Here we provide a stronger version for edge-weighted graphs. Specifically, we prove the following theorem: 

%A straightforward application of Theorem~\ref{thm:mader} to edge-weighted graphs would yield a result that is weaker than what we need to prove the approximation guarantee of our algorithms, as we discuss in the next section. Specifically, let $G=(V,E,w)$ be an edge-weighted graph, let $d$ be a positive real, and assume that $G$ has the density of at least $d$. 
%Now consider an unweighted graph $G'=(V,E)$ defined on the same vertex set $V$ and edge set $E$. As $G'$ has the density of at least $d/\wmax$ (i.e., at least $\lfloor d/\wmax\rfloor$), 
%by Theorem~\ref{thm:mader}, we see that $G'$ has a $\left(\left\lfloor\frac{\lfloor d/\wmax \rfloor}{2}\right\rfloor+1\right)$-vertex-connected subgraph wherein the minimum degree of vertices is greater than $\lfloor d/\wmax\rfloor$. 
%Therefore, we can deduce that $G$ has a $\left(\left\lfloor\frac{\lfloor d/\wmax \rfloor}{2}\right\rfloor+1\right)$-vertex-connected subgraph wherein the minimum weighted degree of vertices is greater than $\wmin \lfloor d/\wmax\rfloor$. 
%
%We provide a stronger version of Mader's theorem on weighted graphs. Specifically, we prove the following theorem.

\begin{theorem}
\label{thm:our_mader}
Let $G=(V,E,w)$ be an edge-weighted graph  and let $d$ be a positive real. 
If $G$ has density at least $d$, then $G$ has a $\left(\left\lfloor \frac{\lceil d/w_\mathrm{max} \rceil}{2}\right\rfloor+1\right)$-vertex-connected subgraph wherein the minimum weighted degree of vertices is greater than $d$. 
\end{theorem}

\begin{proof}
Let $H=(S,E(S))$ be a subgraph of $G$ with the minimum number of vertices that satisfies
\begin{align}\label{ineq:conditions}
\begin{aligned}
|S|\geq \lceil d/\wmax \rceil \quad \text{and} \quad w(S)>d(V)\left(|S|-\frac{\lceil d/\wmax \rceil}{2}\right). 
\end{aligned}
\end{align}
There exists such a subgraph $H$ because $G$ itself satisfies the above condition. 
In fact, since $d(V)\geq d$ holds, there exists a vertex with the weighted degree of at least $2d$, 
implying that the number of neighbors of such a vertex is at least $\lceil 2d/\wmax \rceil$, thus $|V|\geq \lceil 2d/\wmax \rceil + 1 > \lceil d/\wmax \rceil$ holds, and $w(V)=d(V)|V|>d(V)\left(|V|-\frac{\lceil d/\wmax \rceil}{2}\right)$. 
Suppose that $|S|=\lceil d/\wmax \rceil$. Then we have 
\begin{align*}
w(S)
&>d(V)\left(|S|-\frac{\lceil d/\wmax \rceil}{2}\right)
=\frac{d(V)\lceil d/\wmax \rceil}{2}\\
&\geq \frac{\wmax (d/\wmax) \lceil d/\wmax \rceil}{2}
> \wmax{\lceil d/\wmax \rceil \choose 2}=\wmax{|S|\choose 2}\geq w(S), 
\end{align*}
a contradiction. Therefore, we see that $|S|\geq \lceil d/\wmax \rceil+1$. 
Suppose also that there exists a vertex $v$ in $H$ whose weighted degree is at most $d(V)$ in $H$. 
Let $H'=(S', E(S'))$ be a subgraph constructed by removing $v$ from $H$. 
Then we have 
\begin{align*}
|S'|&=|S|-1\geq \lceil d/\wmax \rceil\quad \text{and}\quad \\
w(S')&\geq w(S)-d(V)>d(V)\left(|S|- \frac{\lceil d/\wmax \rceil}{2} -1\right)=d(V)\left(|S'|-\frac{\lceil d/\wmax \rceil}{2}\right). 
\end{align*}
This means that $H'$ also satisfies condition~(\ref{ineq:conditions}), which contradicts the minimality of $H$. 
Therefore, we see that every vertex in $H$ has weighted degree greater than $d(V)\geq d$ in $H$. 

From now on, we show that $H$ is $\left(\left\lfloor \frac{\lceil d/\wmax \rceil}{2}\right\rfloor +1\right)$-vertex-connected. 
Suppose, for contradiction, that there exists $T\subseteq S$ with $|T|\leq \left\lfloor \frac{\lceil d/\wmax \rceil}{2}\right\rfloor$ 
whose removal separates $H$ into two non-empty subgraphs $H[S_1]$ and $H[S_2]$ so that there are no edges between them. 
For any vertex $v\in S_1$, its neighbors in $H$ are all contained in $S_1\cup T$. 
As $v$ has weighted degree greater than $d(V)\geq d$ in $H$, the number of neighbors of $v$ in $S_1\cup T$ is at least $\lceil d/\wmax \rceil$, thus $|S_1\cup T|\geq \lceil d/\wmax \rceil+1$. 
From the minimality of $H$, we see that the subgraph $H[S_1\cup T]$ does not satisfy condition~(\ref{ineq:conditions}), 
which implies that 
\begin{align*}
w(S_1\cup T)\leq d(V)\left(|S_1\cup T|-\frac{\lceil d/\wmax \rceil}{2}\right)
\end{align*}
holds. 
Applying the same argument to $S_2$, we also have
\begin{align*}
w(S_2\cup T)\leq d(V)\left(|S_2\cup T|-\frac{\lceil d/\wmax \rceil}{2}\right). 
\end{align*}
Combining these two inequalities, we have 
\begin{align*}
w(S)&\leq w(S_1\cup T)+w(S_2\cup T)\\
&\leq d(V)(|S_1\cup T|+|S_2\cup T|-\lceil d/\wmax \rceil)\\
&=d(V)(|S_1|+|T|+|S_2|+|T|-\lceil d/\wmax \rceil)\\
&\leq d(V)\left(|S|-\frac{\lceil d/\wmax \rceil}{2}\right), 
\end{align*}
%where the last inequality follows from the assumption $|T|\leq \left\lfloor \frac{\lceil d/\wmax \rceil}{2}\right\rfloor$. 
which contradicts that $H$ satisfies condition~(\ref{ineq:conditions}). 
\end{proof}

%\spara{Finding Mader's subgraph in polynomial time.} 
\subsection{Algorithm for Finding a Mader Subgraph}
We design an algorithm  \texttt{Mader\_subgraph} that extracts a Mader subgraph, i.e., the subgraph whose existence is guaranteed by Theorem~\ref{thm:our_mader}.  To this end, we first present a simple subprocedure, which we call \texttt{Peel}. For an edge-weighted graph $G=(V,E,w)$ and a positive real $d$, 
the procedure $\texttt{Peel}$ returns the maximal subgraph of $G$ 
wherein the minimum weighted degree of vertices is greater than $d$ if such a subgraph exists 
and \texttt{Null} otherwise. 
Specifically, $\texttt{Peel}$ iteratively removes a vertex with the minimum weighted degree in the currently remaining graph 
while the minimum weighted degree is no greater than $d$. 
Note that this procedure is similar to the procedure to find a \emph{$k$-core}. 
For reference, we describe the entire procedure in Algorithm~\ref{alg:peel}, 
where $\mathrm{deg}_S(v)$ for $S\subseteq V$ and $v\in S$ denotes the weighted degree of $v$ in $G[S]$. 
This algorithm can be implemented to run in $O(|E|+|V|\log |V|)$ time, as mentioned in the literature~\cite{Miyauchi+15}. 
\begin{algorithm}[t]
\caption{\texttt{Peel$(G,d)$}}\label{alg:peel}
\SetKwInOut{Input}{Input}
\SetKwInOut{Output}{Output}
\Input{\ $G=(V,E,w)$ and $d\in \mathbb{R}_{>0}$}%, $\tau \in \mathbb{Z}_{>0}$, $d\in \mathbb{R}_{>0}$}
\Output{\ Subgraph of $G$ or \texttt{Null}}
$S\leftarrow V$\;
\While{True}{
$v_\text{min}\leftarrow \argmin_{v\in S}\text{deg}_S(v)$\;
\If{$\mathrm{deg}_S(v_\mathrm{min})> d$}{\Return $G[S]$\;}
$S\leftarrow S\setminus \{v_\text{min}\}$\;
}
\Return{\texttt{Null}}\;
\end{algorithm}

Using Algorithm~\ref{alg:peel}, we present \texttt{Mader\_subgraph} in Algorithm~\ref{alg:mader}, 
where the notation $V(H')$ denotes the vertex set of subgraph $H'$ of $G$. 
%Our algorithm keeps a family of subgraphs that potentially contain $G^*$ as a subgraph. 
Here we briefly explain the behavior of the algorithm. 
Let $G^*$ be a Mader subgraph of a given edge-weighted graph $G$. 
The algorithm keeps a family of subgraphs $\mathcal{H}$ in which exactly one subgraph contains $G^*$ as its subgraph. 
In each iteration, the algorithm tests whether a subgraph in $\mathcal{H}$ is a Mader subgraph or not, 
and if not, the algorithm divides the subgraph into strictly smaller pieces and add (a part of) them to $\mathcal{H}$. 
The algorithm repeats this operation until it finds a Mader subgraph. 
It should be noted that our algorithm is based on Matula's algorithm~\cite[Algorithm~A]{Matula78}, 
which finds the most highly connected subgraph in terms of vertex connectivity, 
i.e., $H\in \argmax\{\kappa(H)\mid \text{$H$ is a subgraph of $G$}\}$. 
\begin{algorithm}[t]
\caption{\texttt{Mader\_subgraph$(G)$}}\label{alg:mader}
\SetKwInOut{Input}{Input}
\SetKwInOut{Output}{Output}
\Input{\ $G=(V,E,w)$}%, $\tau \in \mathbb{Z}_{>0}$, $d\in \mathbb{R}_{>0}$}
\Output{\ Subgraph of $G$}
$H\leftarrow \texttt{Peel}(G, d(V))$\;
$\tau\leftarrow \left\lfloor \frac{\lceil d(V)/\wmax \rceil}{2}\right\rfloor +1$\;
$\mathcal{H} \leftarrow \text{the family of the maximal connected subgraphs of $H$ that have at least $\tau+1$ vertices}$\;
\If{there exists a clique $K$ in $\mathcal{H}$}{\Return{$K$}\;}
\While{True}{
$H'\leftarrow \text{an arbitrary element of } \mathcal{H}$\;
%$H'\leftarrow \argmax_{H\in \mathcal{H}}|V(H)|$\;
$C\leftarrow \text{the minimum vertex separator of $H'$}$\;
\If{$|C|\geq \tau$}{\Return{$H'$}\;}
$\mathcal{S}\leftarrow \text{the family of the vertex sets of the maximal connected subgraphs of $G[V(H')\setminus C]$}$\;
%$\mathcal{S}\leftarrow \text{the family of the vertex sets of the connected subgraphs of $H'\setminus C$}$\;
%$\{S_1,S_2,\dots S_q\}\leftarrow \text{the family of the connected components of $H'\setminus C$}$\;
$\mathcal{H}'\leftarrow \emptyset$\;
\For{each $S\in \mathcal{S}$}{
\If{$\texttt{Peel}(G[S\cup C],d(V))$ has at least $\tau+1$ vertices}{$\mathcal{H'}\leftarrow \mathcal{H'}\cup \{\texttt{Peel}(G[S\cup C],d(V))\}$\;}
}
\If{there exists a clique $K$ in $\mathcal{H}'$}{\Return{K}\;}
$\mathcal{H} \leftarrow (\mathcal{H}\setminus \{H'\}) \cup \mathcal{H'}$\;
}
\end{algorithm}

The following theorem verifies the validity of \texttt{Mader\_subgraph}.  The proof strategy is similar to that for Matula~\cite[Theorem~3]{Matula78}.  
%The following theorem plays a key role in the proof of our bicriteria approximation algorithm for Problem~\ref{prob:vertex} in Section~\ref{sec:biapprox}.  The proof strategy is similar to that for Theorem~3 in Matula~\cite{Matula78}.  

\begin{theorem}
For a given edge-weighted graph $G=(V,E,w)$, Algorithm~\ref{alg:mader} outputs a Mader subgraph of $G$ in $O(|V|^{19/4})$ time. 
\end{theorem}

\begin{proof}
It is easy to see that if the algorithm terminates, its output is a Mader subgraph of $G$. 
Thus, in what follows, we analyze the time complexity of the algorithm. 

Specifically, we prove that Algorithm~\ref{alg:mader} runs in $O(|V|^{19/4})$ time. 
The time complexity of the algorithm except for the while-loop is given by $O(|E|+|V|\log|V|)$ 
due to the time complexity of the procedure \texttt{Peel}. 
We can show that the time complexity of the while-loop is given by $O(|V|^{19/4})$. 
To see this, we analyze the time complexity of each iteration and the number of iterations. 
The time complexity of each iteration is dominated by that required to compute the minimum vertex separator $C$ of $H'$. 
As reviewed in Section~\ref{sec:related}, the current best algorithm completes this task in $O(|V(H')|(|C|^2\cdot \min\{|V(H')|^{3/4}, |C|^{3/2}\}+|C| |V(H')|))$ time. 
Hence, the time complexity of each iteration is bounded by $O(|V|^{15/4})$. 
Next we show that the number of iterations of the while-loop is bounded by $|V|$. 
Let $G^*$ be a Mader subgraph of $G$, that is, 
$G^*$ is a $\tau$-vertex-connected subgraph of $G$ wherein the minimum weighted degree of vertices is greater than $d(V)$. 
It is easy to see that exactly one subgraph in $\mathcal{H}$ contains $G^*$ as its subgraph in any iteration of the while-loop. 
Here we define the \emph{surplus} of $\mathcal{H}$ as
\begin{align*}
s(\mathcal{H})=\sum_{H\in \mathcal{H}}(|V(H)|-\tau-1). 
\end{align*}
For the initial $\mathcal{H}$, we have $s(\mathcal{H})\leq |V|-\tau -1$. 
Note that $s(\mathcal{H})\geq 0$ holds in any iteration. 
Let us consider an arbitrary iteration in which the algorithm does not terminate. 
Let $\mathcal{S}'=\{S\in \mathcal{S}\mid |V(\texttt{Peel}(G[S\cup C],d(V)))|\geq \tau +1\}$. 
If $|\mathcal{S}'|\leq 1$ holds, then $H'$ is simply deleted or replaced by a subgraph with at most $|V(H')|-1$ vertices, 
in the updated $\mathcal{H}$. 
Thus, the surplus decreases by at least one in the iteration.  
Assume that $|\mathcal{S}'|\geq 2$. 
Then we have 
\begin{align*}
\sum_{H\in \mathcal{H'}}(|V(H)|-\tau-1)
&=\sum_{S\in \mathcal{S}'}(|V(\texttt{Peel}(G[S\cup C],d(V)))|-\tau-1)\\
&\leq \sum_{S\in \mathcal{S}'}(|V(G[S\cup C])|-\tau-1)\\
&\leq |V(H')|+(|\mathcal{S}'|-1)(|C|-\tau)-\tau-|\mathcal{S}'|\\
&<|V(H')|-\tau-2, 
\end{align*}
where the last inequality follows from $|\mathcal{S}'|\geq 2$ and $|C|<\tau$. 
Note that $|C|<\tau$ holds because the algorithm has not yet terminated in the iteration. 
The above inequality implies that the surplus decreases by at least two in the iteration. 
Therefore, the number of iterations of the while-loop is bounded by $|V|-\tau < |V|$. 
\end{proof}

\section{Bicriteria Approximation Algorithms}
\label{sec:biapprox} 

%Unlike the above approximation algorithm, we here consider an algorithm for the general edge-weighted case. 
In this section, we first design a polynomial-time $\left(\frac{\gamma}{4}\cdot \frac{\wmin}{\wmax},1/\gamma\right)$-bicriteria approximation algorithm with parameter $\gamma \in [1,2]$ for Problem~\ref{prob:vertex}, and then present a corresponding result for Problem~\ref{prob:edge}.

\subsection{Algorithm for Problem~\ref{prob:vertex}}
%The algorithm %is similar to Algorithm~\ref{alg:approximation}, but it 
%has an additional parameter $\gamma \in [1,2]$, enabling us to adjust the balance of the bicriteria approximation ratio. 
For a given edge-weighted graph $G=(V,E,w)$, our algorithm first finds the family of maximal $k$-vertex-connected subgraphs 
$\{G[S_1],\dots, G[S_p]\}$ 
using Makino's algorithm~\cite{Makino88} combined with Gabow's vertex connectivity algorithm~\cite{Gabow06}, 
which takes $O(|V|^2(k^2\cdot \min\{|V|^{3/4}, k^{3/2}\}+k |V|))$ time.  
Note that if there is no $k$-vertex-connected subgraph found, 
our algorithm returns \texttt{INFEASIBLE} because the instance is actually infeasible. 

For each $i=1,\dots, p$, the algorithm initializes $S^*_i$ as $S_i$. 
Then the algorithm finds a densest subgraph $S^\text{DS}_i$ (without any constraint) in $G[S_i]$. 
This can be done in polynomial time using Charikar's linear-programming-based algorithm for the densest subgraph problem~\cite{Charikar00}. 
After that, if 
$k\leq \gamma \left(\left\lfloor \frac{\lceil d(S^\text{DS}_i)/\wmax \rceil}{2}\right\rfloor +1\right)$ holds, 
then the algorithm employs as $S^*_i$ 
the vertex set of a Mader subgraph of $G[S^\text{DS}_i]$, i.e., the vertex set of a $\left(\left\lfloor \frac{\lceil d(S^\text{DS}_i)/\wmax(G[S^\text{DS}_i]) \rceil}{2}\right\rfloor +1\right)$-vertex-connected subgraph 
in $G[S^\text{DS}_i]$ wherein the minimum weighted degree of vertices is greater than $d(S^\text{DS}_i)$, using the procedure \texttt{Mader\_subgraph} (Algorithm~\ref{alg:mader}). 
Here $\wmax(G[S^\text{DS}_i])$ denotes the maximum weight of edges in $G[S^\text{DS}_i]$. 
Note that $\wmax(G[S^\text{DS}_i])\leq \wmax$ holds. 
%The existence of such a subgraph is guaranteed by Theorem~\ref{thm:our_mader}.  
For $G[S^\text{DS}_i]$, \texttt{Mader\_subgraph} runs in $O(|S^\text{DS}_i|^{19/4})=O(|V|^{19/4})$ time.

Finally, the algorithm outputs the densest subset among $\{S^*_1,\dots, S^*_p\}$. 
For reference, we summarize the entire procedure in Algorithm~\ref{alg:biapprox_vertex}. 
As the maximum total number of maximal $k$-vertex-connected subgraphs is $O(|V|)$~\cite{Matula77}, 
the overall running time of Algorithm~\ref{alg:biapprox_vertex} is given by $O(|V|(|V|^{19/4}+T_\text{DS}))$, 
%$O(|V|(|V|^3(\Delta^2\cdot \min\{|V|^{3/4}, \Delta^{3/2}\}+\Delta |V|) +T))$, 
where $T_\text{DS}$ is the computation time required to find a densest subgraph in (any subgraph of) $G$. 
Note that as mentioned above, $T_\text{DS}$ is polynomial in $|V|$ and $|E|$. 
Moreover, for unweighted graphs, Goldberg's flow-based algorithm~\cite{Goldberg84} gives $T_\text{DS}=O(|E||V|)$, using Orlin's maximum-flow algorithm \cite{orlin2013max}. 

\begin{algorithm}[t]
\caption{Bicriteria approximation algorithm with parameter $\gamma \in [1,2]$ for Problem~\ref{prob:vertex}}
\label{alg:biapprox_vertex}
\SetKwInOut{Input}{Input}
\SetKwInOut{Output}{Output}
\Input{\ $G=(V,E,w)$ and $k\in \mathbb{Z}_{>0}$}
\Output{\ $S\subseteq V$ or \texttt{INFEASIBLE}}
Find the family of maximal $k$-vertex-connected subgraphs $\{G[S_1],\dots, G[S_p]\}$\;
\If{there is no $k$-vertex-connected subgraph found}{
\Return \texttt{INFEASIBLE}\;
}
\Else{
\For{$i=1,\dots, p$}{
$S^*_i\leftarrow S_i$\;
Find a densest subgraph $S^\text{DS}_i$ (without any constraint) in $G[S_i]$\;
\If{$k\leq \gamma \left(\left\lfloor \frac{\lceil d(S^\mathrm{DS}_i)/\wmax \rceil}{2}\right\rfloor +1\right)$}{
%\If{$k\leq \gamma(d(S^\text{DS}_i)/2+1)$}{
$S^*_i \leftarrow \text{The vertex set of } \texttt{Mader\_subgraph}(G[S^\text{DS}_i])$\; 
%(i.e., $S^*_i \leftarrow$ The vertex set of a $\left(\left\lfloor \frac{\lceil d(S^\text{DS}_i)/\wmax(G[S^\text{DS}_i]) \rceil}{2}\right\rfloor +1\right)$-vertex-connected subgraph in $G[S^\text{DS}_i]$ wherein the minimum weighted degree of vertices is greater than $d(S^\text{DS}_i)$)
%$S^*_i \leftarrow$ A $\wmin \left(\left\lfloor \frac{\lceil d(S^\text{DS}_i)/\wmax \rceil}{2}\right\rfloor +1\right)$-edge-connected subgraph in $G[S^\text{DS}_i]$ wherein the minimum degree of vertices is greater than $d(S^\text{DS}_i)$\;
}
}
\Return $S\in \argmax_{S\in \{S^*_1,\dots, S^*_p\}}d(S)$\;
}
\end{algorithm}

\subsection{Analysis}
Using our generalized Mader's theorem (Theorem~\ref{thm:our_mader}), we provide the bicriteria approximation ratio of Algorithm~\ref{alg:biapprox_vertex}: 

\begin{theorem}\label{thm:biapprox_vertex}
For any $\gamma\in [1,2]$, Algorithm~\ref{alg:biapprox_vertex} is a polynomial-time $\left(\frac{\gamma}{4}\cdot \frac{\wmin}{\wmax},1/\gamma\right)$-bicriteria approximation algorithm 
%For any $\gamma\in [1,2]$, Algorithm~\ref{alg:biapproximation_weighted} is a polynomial-time $\left(\frac{4\cdot \wmax}{\gamma\cdot \wmin},\gamma\right)$-bicriteria approximation algorithm 
for Problem~\ref{prob:vertex}. 
\end{theorem}

\begin{proof}
We first show that the output of Algorithm~\ref{alg:biapprox_vertex} is $(k/\gamma)$-vertex-connected. 
To this end, it suffices to confirm $(k/\gamma)$-vertex-connectivity of $G[S^*_i]$ for each $i=1,\dots, p$. 
Fix $i\in \{1,\dots, p\}$. 
If $k\leq \gamma \left(\left\lfloor \frac{\lceil d(S^\text{DS}_i)/\wmax \rceil}{2}\right\rfloor +1\right)$ does not hold, 
we are done since $G[S^*_i]$ is given by $G[S_i]$, which is $k$-vertex-connected 
(thus $(k/\gamma)$-vertex-connected). 
Consider the case where $k\leq \gamma \left(\left\lfloor \frac{\lceil d(S^\text{DS}_i)/\wmax \rceil}{2}\right\rfloor +1\right)$ holds. 
Applying Theorem~\ref{thm:our_mader} to $G[S^\text{DS}_i]$ with setting $d=d(S^\text{DS}_i)$, 
we see that $G[S^\text{DS}_i]$ has a $\left(\left\lfloor \frac{\lceil d(S^\text{DS}_i)/w_\mathrm{max}(G[S^\text{DS}_i]) \rceil}{2}\right\rfloor +1\right)$-vertex-connected subgraph, 
%we see that $G[S^\text{DS}_i]$ has a $\left(\left\lfloor \frac{\lceil d(S^\text{DS}_i)/w_\mathrm{max} \rceil}{2}\right\rfloor +1\right)$-vertex-connected subgraph, 
which is $(k/\gamma)$-vertex-connected. 
Algorithm~\ref{alg:biapprox_vertex} employs such a subset as $S^*_i$. 

We next analyze the first term of the bicriteria approximation ratio. 
%As in the proof of Theorem~\ref{thm:approximation}, 
It suffices to show that for each $i=1,\dots, p$, the subset $S^*_i$ has density at least $\frac{\gamma}{4}\cdot\frac{\wmin}{\wmax}$ times the optimal value of Problem~\ref{prob:vertex} on $G[S_i]$. 
Fix $i\in \{1,\dots, p\}$. 
%Recall that $S^\text{DS}_i$ denotes a densest subgraph (without any constraint) in $G[S_i]$. 
Clearly, the optimal value of Problem~\ref{prob:vertex} on $G[S_i]$, which we denote by $\textsf{OPT}_i$, 
is at most $d(S^\text{DS}_i)$. 

We first consider the case where $k\leq \gamma \left(\left\lfloor \frac{\lceil d(S^\text{DS}_i)/\wmax \rceil}{2}\right\rfloor +1\right)$ does not hold. 
In this case, Algorithm~\ref{alg:biapprox_vertex} just employs $S_i$ as $S^*_i$. 
As $G[S_i]$ is $k$-vertex-connected, each vertex has weighted degree of at least $\wmin \cdot k > \gamma \cdot \wmin \left(\left\lfloor \frac{\lceil d(S^\text{DS}_i)/\wmax \rceil}{2}\right\rfloor +1\right)$; 
thus, the density of $S_i$ is greater than 
\begin{align*}
\gamma \cdot \wmin \left(\left\lfloor \frac{\lceil d(S^\text{DS}_i)/\wmax \rceil}{2}\right\rfloor +1\right)/2 
&\geq \gamma \cdot \wmin \left(\frac{d(S^\text{DS}_i)/\wmax }{2}-\frac{1}{2}+1\right)/2\\
%&\geq \gamma \cdot \wmin \left(\frac{d(S^\text{DS}_i)/\wmax }{2}-\frac{1}{2}+1\right)/2\\
&> \frac{\gamma}{4}\cdot\frac{\wmin}{\wmax}\cdot d(S^\text{DS}_i)
\geq \frac{\gamma}{4}\cdot\frac{\wmin}{\wmax}\cdot \textsf{OPT}_i, 
\end{align*}
which means $\frac{\gamma}{4}\cdot \frac{\wmin}{\wmax}$-approximation. 

We next consider the case where $k\leq \gamma \left(\left\lfloor \frac{\lceil d(S^\text{DS}_i)/\wmax \rceil}{2}\right\rfloor +1\right)$ holds. 
Applying Theorem~\ref{thm:our_mader} to $G[S^\text{DS}_i]$ with setting $d=d(S^\text{DS}_i)$, 
we see that $G[S^\text{DS}_i]$ has a $\left(\left\lfloor \frac{\lceil d(S^\text{DS}_i)/\wmax(G[S^\text{DS}_i]) \rceil}{2}\right\rfloor +1\right)$-vertex-connected subgraph 
wherein the minimum weighted degree of vertices is greater than $d(S^\text{DS}_i)$. 
Algorithm~\ref{alg:biapprox_vertex} employs such a subset as $S^*_i$. 
As each vertex has weighted degree greater than $d(S^\text{DS}_i)$, 
the density of $S^*_i$ is greater than $d(S^\text{DS}_i)/2 \geq \textsf{OPT}_i/2$, 
which means $1/2$-approximation (thus $\frac{\gamma}{4}\cdot \frac{\wmin}{\wmax}$-approximation). 
\end{proof}

From the proof, we see that if the if-condition of Algorithm~\ref{alg:biapprox_vertex} holds, the output admits $1/2$-approximation, 
irrespective of edge weights. 
Moreover, it should be noted that setting $\gamma=1$ in the theorem, 
we can obtain an ordinary $\frac{1}{4}\cdot \frac{\wmin}{\wmax}$-approximation algorithm for Problem~\ref{prob:vertex}. 
In Section~\ref{sec:unweighted}, we present an algorithm with a better approximation ratio. 
%Setting $\gamma=1$ in Theorem~\ref{thm:main}, we have the following corollary. 
%\begin{corollary}
%There is a polynomial-time 4-approximation algorithm for Problem~\ref{prob:original}. 
%\end{corollary}

\subsection{Algorithm for Problem~\ref{prob:edge} and Analysis}
Here we present a bicriteria approximation algorithm for Problem~\ref{prob:edge}, 
which is an edge-connectivity counterpart of Algorithm~\ref{alg:biapprox_vertex}. 
For a given edge-weighted graph $G=(V,E,w)$, our algorithm first finds the family of maximal $k$-edge-connected subgraphs 
$\{G[S_1],\dots, G[S_p]\}$. 
As reviewed in Section~\ref{sec:related}, this can be done in $O(|V|^2(|E|+|V|\log|V|))$ time 
using one of the minimum cut algorithms by Nagamochi and Ibaraki~\cite{Nagamochi_Ibaraki_92}, 
Stoer and Wagner~\cite{Stoer_Wagner_97}, and Frank~\cite{Frank94} as a subroutine. 
%For $S\subseteq V$ and $k\in \mathbb{R}_{>0}$, 
%the induced subgraph $G[S]$ is called a \emph{maximal $k$-edge-connected subgraph}
%if $G[S]$ is $k$-edge-connected and no superset of $S$ has this property. 
%Note that unlike the vertex-connectivity case, maximal $k$-edge-connected subgraphs do not overlap; 
%therefore, it is obvious that the maximum total number of maximal $k$-edge-connected subgraphs is linear in the number of vertices of the graph. 
%Specifically, to compute the family, our algorithm performs the following naive operation: 
%as long as the minimum cut value of a (sub)graph at hand is less than $k$, 
%divide the graph into two subgraphs along with the cut, recursively, 
%and collect the resulting $k$-edge-connected subgraphs. 
%If we use one of the minimum cut algorithms by Nagamochi and Ibaraki~\cite{Nagamochi_Ibaraki_92} and Stoer and Wagner~\cite{Stoer_Wagner_97}, both of which run in $O(|V|(|E|+|V|\log|V|))$ time, 
%the above naive operation can be done in $O(|V|^2(|E|+|V|\log|V|))$ time. 
If $G$ is simple unweighted, the time complexity reduces to $O(|E||V|\log^2|V|\log\log^2|V|)$ 
using the minimum cut algorithm by Henzinger et al.~\cite{Henzinger20}. 
%{\color{red} 
%[Can we improve this?] 
%}

In the processing of $G[S_i]$ for each $i=1,\dots, p$, 
the algorithm computes a \emph{variant} of a Mader subgraph of $G[S^\text{DS}_i]$, 
i.e., a $\wmin \left(\left\lfloor \frac{\lceil d(S^\text{DS}_i)/\wmax \rceil}{2}\right\rfloor +1\right)$-edge-connected subgraph 
in $G[S^\text{DS}_i]$ wherein the minimum weighted degree of vertices is greater than $d(S^\text{DS}_i)$. 
The existence of such a subgraph is guaranteed by a corollary of Theorem~\ref{thm:our_mader}, which we will present later. Recall that Algorithm~\ref{alg:biapprox_vertex} uses the procedure \texttt{Mader\_subgraph}. 
On the other hand, the above variant can be computed 
using the strategy employed by the algorithms for computing the family of maximal $k$-edge-connected subgraphs, 
presented in Section~\ref{sec:related}.
Specifically, the strategy in our scenario is as follows: 
if the weight of the minimum cut of $G[S^\text{DS}_i]$ is less than $\wmin \left\lfloor \frac{\lceil d(S^\text{DS}_i)/\wmax \rceil}{2}\right\rfloor +1$, divide the graph into two subgraphs along with the cut and then repeat the procedure on the resulting subgraphs 
(until it finds the variant of a Mader subgraph). 
It should be noted that in order to satisfy the minimum weighted degree condition, 
our algorithm needs to conduct the procedure \texttt{Peel} every time before it processes a new subgraph. 
For reference, the pseudocode of our algorithm is given in Algorithm~\ref{alg:biapprox_edge}. 

Here we evaluate the running time of Algorithm~\ref{alg:biapprox_edge}. 
It is easy to see that the above algorithm for finding the variant of a Mader subgraph still has the same running time as 
that of algorithms for computing the family of maximal $k$-edge-connected subgraphs. 
Therefore, the time complexity of the processing of each $G[S_i]$ is bounded by $O(T_\text{DS}(S_i)+|S_i|^2(|E(S_i)|+|S_i|\log |S_i|))$, 
where $T_\text{DS}(S_i)$ is the computation time required to find a densest subgraph in $G[S_i]$. 
%For simple unweighted graphs, the time complexity can be replaced by $O(|E(S_i)||S_i|\log^2|S_i|\log\log^2|S_i|+|S_i|^3)$. 
Recalling that maximal $k$-edge-connected subgraphs do not overlap for any $k$, 
we see that the time complexity of the entire for-loop is bounded by $O(T_\text{DS}(G)+|V|^2(|E|+|V|\log |V|))$, 
which also bounds the overall running time of Algorithm~\ref{alg:biapprox_edge}. 
For simple unweighted graphs, we have the running time of $O(|V|^3+|E||V|\log^2|V|\log\log^2|V|)$. 
%Note that if there is no $k$-edge-connected subgraph found, the algorithm returns \texttt{INFEASIBLE} because the instance is infeasible. 

%Note that for unweighted graphs, there is an algorithm that computes a family of maximal $k$-edge-connected subgraphs in $O(mn\textsf{polylog}()$ time. 
%\atsushi{Which one is the fastest minimum cut algorithm for undirected, weighted graphs?}

\begin{algorithm}[t]
\caption{Bicriteria approximation algorithm with parameter $\gamma \in [1,2]$ for Problem~\ref{prob:edge}}
\label{alg:biapprox_edge}
\SetKwInOut{Input}{Input}
\SetKwInOut{Output}{Output}
\Input{\ $G=(V,E,w)$ and $k\in \mathbb{R}_{>0}$}
\Output{\ $S\subseteq V$ or \texttt{INFEASIBLE}}
Find the family of maximal $k$-edge-connected subgraphs $\{G[S_1],\dots, G[S_p]\}$\;
\If{there is no $k$-edge-connected subgraph found}{
\Return \texttt{INFEASIBLE}\;
}
\Else{
\For{$i=1,\dots, p$}{
$S^*_i\leftarrow S_i$\;
Find a densest subgraph $S^\text{DS}_i$ (without any constraint) in $G[S_i]$\;
\If{$k\leq \gamma \cdot \wmin \left(\left\lfloor \frac{\lceil d(S^\text{DS}_i)/\wmax \rceil}{2}\right\rfloor +1\right)$}{
%\If{$k\leq \gamma(d(S^\text{DS}_i)/2+1)$}{
$S^*_i \leftarrow$ The vertex set of a $\wmin \left(\left\lfloor \frac{\lceil d(S^\text{DS}_i)/\wmax \rceil}{2}\right\rfloor +1\right)$-edge-connected subgraph in $G[S^\text{DS}_i]$ wherein the minimum weighted degree of vertices is greater than $d(S^\text{DS}_i)$\;
%$S^*_i \leftarrow$ A $\wmin \left(\left\lfloor \frac{\lceil d(S^\text{DS}_i)/\wmax \rceil}{2}\right\rfloor +1\right)$-edge-connected subgraph in $G[S^\text{DS}_i]$ wherein the minimum degree of vertices is greater than $d(S^\text{DS}_i)$\;
}
}
\Return $S\in \argmax_{S\in \{S^*_1,\dots, S^*_p\}}d(S)$\;
}
\end{algorithm}

%Let us analyze the running time of the algorithm. 
%We can compute a family of maximal $k$-edge-connected subgraphs in $O(|E||V|^2+|V|^3\log |V|)$ time 
%using the minimum cut algorithm by Stoer and Wagner~\cite{Stoer_Wagner_97} as a subroutine. 
%A densest subgraph 
%The overall running time is therefore {\color{red} $O(|E||V|^2+|V|^3\log |V|)$}. 

Finally we analyze the theoretical performance guarantee of Algorithm~\ref{alg:biapprox_edge}. 
It is easy to see that any (edge-weighted) $k$-vertex-connected graph $G$ is $\wmin k$-edge-connected, 
which gives the following corollary to Theorem~\ref{thm:our_mader}: 

\begin{corollary}\label{cor:our_mader}
Let $G=(V,E,w)$ be an edge-weighted graph and let $d$ be a positive real. 
If $G$ has density at least $d$, 
then $G$ has a $\wmin \left(\left\lfloor \frac{\lceil d/w_\mathrm{max} \rceil}{2}\right\rfloor +1\right)$-edge-connected subgraph wherein the minimum weighted degree of vertices is greater than $d$. 
\end{corollary}

Using this corollary, we can derive the bicriteria approximation ratio of Algorithm~\ref{alg:biapprox_edge}: 

\begin{theorem}\label{thm:biapprox_edge}
For any $\gamma\in [1,2]$, Algorithm~\ref{alg:biapprox_edge} is a polynomial-time $\left(\frac{\gamma}{4}\cdot \frac{\wmin}{\wmax},1/\gamma\right)$-bicriteria approximation algorithm 
%For any $\gamma\in [1,2]$, Algorithm~\ref{alg:biapproximation_weighted} is a polynomial-time $\left(\frac{4\cdot \wmax}{\gamma\cdot \wmin},\gamma\right)$-bicriteria approximation algorithm 
for Problem~\ref{prob:edge}. 
\end{theorem}

The proof is similar to that of Theorem~\ref{thm:biapprox_vertex}, and is omitted.

\subsection{Remarks on Theorem~\ref{thm:our_mader}}\label{subsec:remark}
Here we explain that our generalized Mader's theorem (i.e., Theorem~\ref{thm:our_mader}) is essential to derive the bicriteria approximation ratio given in Theorems~\ref{thm:biapprox_vertex} and \ref{thm:biapprox_edge}. To this end, recall that the straightforward application of the original Mader's theorem to edge-weighted graphs derives the following statement: Let $G=(V,E,w)$ be an edge-weighted graph and let $d$ be a positive real. If $G$ has density at least $d$, then $G$ has a $\left(\left\lfloor\frac{\lfloor d/\wmax \rfloor}{2}\right\rfloor+1\right)$-vertex-connected subgraph wherein the minimum weighted degree of vertices is greater than $\wmin \lfloor d/\wmax\rfloor$.

%Here we explain that our generalized Mader's theorem (i.e., Theorem~\ref{thm:our_mader}) 
%cannot be directly obtained from the original one (i.e., Theorem~\ref{thm:mader}), 
%and our theorem is essential to derive the bicriteria approximation ratio given in Theorems~\ref{thm:biapprox_vertex} and \ref{thm:biapprox_edge}. 

%As in Theorem~\ref{thm:our_mader}, let $G=(V,E,w)$ be an edge-weighted graph, let $d$ be a positive real, 
%and assume that $G$ has the density of at least $d$. 
%Now consider an unweighted graph $G'=(V,E)$ defined on the same vertex set $V$ and edge set $E$. 
%As $G'$ has the density of at least $d/\wmax$ (i.e., at least $\lfloor d/\wmax\rfloor$), 
%by Theorem~\ref{thm:mader}, we see that $G'$ has a $\left(\left\lfloor\frac{\lfloor d/\wmax \rfloor}{2}\right\rfloor+1\right)$-vertex-connected subgraph wherein the minimum degree of vertices is greater than $\lfloor d/\wmax\rfloor$. 
%Therefore, we can deduce that $G$ has a $\left(\left\lfloor\frac{\lfloor d/\wmax \rfloor}{2}\right\rfloor+1\right)$-vertex-connected subgraph wherein the minimum weighted degree of vertices is greater than $\wmin \lfloor d/\wmax\rfloor$. 

Obviously, the above statement is weaker than Theorem~\ref{thm:our_mader}. 
Indeed, vertex connectivity of $\left\lfloor\frac{\lceil d/\wmax \rceil}{2}\right\rfloor+1$ in Theorem~\ref{thm:our_mader} has decreased to $\left\lfloor\frac{\lfloor d/\wmax \rfloor}{2}\right\rfloor+1$, which is only a slight deterioration, 
but the minimum weighted degree of $d$ in Theorem~\ref{thm:our_mader} has significantly decreased to $\wmin \lfloor d/\wmax\rfloor$. 
It is easy to see that to prove Theorems~\ref{thm:biapprox_vertex} and \ref{thm:biapprox_edge}, 
vertex connectivity of $\left\lfloor\frac{\lfloor d/\wmax \rfloor}{2}\right\rfloor+1$ is sufficient, 
but the minimum weighted degree of $\wmin \lfloor d/\wmax\rfloor$ is insufficient. 
In fact, in the last paragraph of the proof of Theorem~\ref{thm:biapprox_vertex}, 
by using the decreased minimum weighted degree, 
we can only guarantee that the density of $S_i^*$ is greater than 
$\frac{\wmin \lfloor d(S^\text{DS}_i)/\wmax\rfloor}{2}\geq \frac{\wmin \lfloor \OPT_i/\wmax\rfloor}{2}$ 
(rather than $d(S^\text{DS}_i)/2\geq \OPT_i/2$ in the proof). 
Note that $\frac{\wmin \lfloor \OPT_i/\wmax\rfloor}{2}$ may be less than $\frac{\gamma}{4}\cdot \frac{\wmin}{\wmax}\cdot \OPT_i$, 
meaning that the decreased minimum weighted degree is insufficient to prove the theorem. 
We can see the same issue in the proof of Theorem~\ref{thm:biapprox_edge}.

\section{Approximation Algorithms}
\label{sec:unweighted} 
In this section, we design a polynomial-time $\left(\frac{6}{19}\cdot \frac{\wmin}{\wmax}\right)$-approximation algorithm for Problem~\ref{prob:vertex}, which improves the approximation ratio of   $\frac{1}{4}\cdot \frac{\wmin}{\wmax}$ that is immediately derived by  Algorithm~\ref{alg:biapprox_vertex}. Then we present its counterpart result for Problem~\ref{prob:edge}.  

\subsection{Algorithm for Problem~\ref{prob:vertex}} 
Our algorithm first computes the most highly connected subgraph in terms of vertex connectivity, 
i.e., $H\in \argmax\{\kappa(H)\mid \text{$H$ is a subgraph of $G$}\}$.   
This can be done using Matula's algorithm~\cite[Algorithm~A]{Matula78}.
Then our algorithm simply returns the subgraph if its vertex connectivity is no less than $k$ and \texttt{INFEASIBLE} otherwise.  Our algorithm is described in pseudocode as Algorithm~\ref{alg:approx_vertex}. 

Matula~\cite{Matula78} showed that the time complexity of the algorithm for computing the most highly connected subgraph in terms of vertex connectivity is given by $O(|V|\cdot T)$, 
where $T$ is the computation time required to find a minimum vertex separator of $G$. 
If we consider Gabow's vertex connectivity algorithm~\cite{Gabow06}, 
the time complexity becomes $O(|V|^2(\kappa(G)^2\cdot \min\{|V|^{3/4}, \kappa(G)^{3/2}\}+\kappa(G) |V|))$. 
%where $\Delta$ denotes the maximum (unweighted) degree of vertices in $G$, a trivial upper bound on the size of a minimum vertex separator of $G$. 
Clearly, Algorithm~\ref{alg:approx_vertex} has the same time complexity. 

\begin{algorithm}[t]
\caption{Approximation algorithm for Problem~\ref{prob:vertex}}\label{alg:approx_vertex}
%\caption{Approximation Algorithm for Problem~\ref{prob:special} with Unweighted Case}\label{alg:approximation}
\SetKwInOut{Input}{Input}
\SetKwInOut{Output}{Output}
\Input{\ $G=(V,E,w)$ and $k\in \mathbb{Z}_{>0}$}
%\Input{\ $G=(V,E)$ and $k\in \mathbb{Z}_{>0}$}
\Output{\ $S\subseteq V$ or \texttt{INFEASIBLE}}
%\Output{\ Either {\sc INFEASIBLE}\ or $S\subseteq V$}
%$\ell \gets  \max\{ t \mid \exists \text{ a non-empty $t$-edge-connected subgraph of } G \}$
$H\leftarrow \argmax\{\kappa(H)\mid \text{$H$ is a subgraph of $G$}\}$\;
\If{$\kappa(H) \geq k$}{
\Return the vertex set of $H$\;
}
\Else{
\Return \texttt{INFEASIBLE}\;
}
\end{algorithm}
 
\subsection{Analysis}
From now on, we analyze the theoretical performance guarantee of Algorithm~\ref{alg:approx_vertex}. 
To this end, we use the following theorem, which is a useful variant of Mader's theorem: 
\begin{theorem}[Bernshteyn and Kostochka~\cite{Bernshteyn_Kostochka_16}]\label{thm:Bernshteyn}
Let $G=(V,E)$ be an unweighted graph and let $t$ be an integer with $t\geq 2$. 
If $G$ satisfies $|V|\geq \frac{5}{2}t$ and $|E|>\frac{19}{12}t(|V|-t)$, then $G$ has a $(t+1)$-vertex-connected subgraph. 
\end{theorem}
 
We provide the approximation ratio of Algorithm~\ref{alg:approx_vertex} in the following theorem: 
\begin{theorem}\label{thm:approx_vertex}
Algorithm~\ref{alg:approx_vertex} is a polynomial-time $\left(\frac{6}{19}\cdot \frac{\wmin}{\wmax}\right)$-approximation algorithm for Problem~\ref{prob:vertex}. 
\end{theorem}

\begin{proof}
Let $S\subseteq V$ be the output of Algorithm~\ref{alg:approx_vertex}.  
Define
\begin{align*}
\kappa_\text{max} =  \max\{\kappa(H)\mid \text{$H$ is a subgraph of $G$} \}.
\end{align*}
As we assumed that $|E|\geq 1$, we have $\kappa_\text{max}\geq 1$. 
Recall that $H=G[S]$ is a $\kappa_\text{max}$-vertex-connected subgraph. 
We denote by \OPT the density of an optimal solution to Problem~\ref{prob:vertex}. 
Let $S_\text{DS}\subseteq V$ be a densest subgraph (unconstrained) in $G$. 
As $d(S_\text{DS})\geq \OPT$, it suffices to show that $d(S)\geq \frac{6}{19}\cdot \frac{\wmin}{\wmax}\cdot d(S_\text{DS})$ holds.
Let $n_\text{DS}$ and $m_\text{DS}$ denote the number of vertices and edges in $G[S_\text{DS}]$, respectively. 

\textbf{Case I: $\kappa_\mathrm{max} = 1$.} In this case, $G$ is a forest; 
therefore, using the fact that $m_\text{DS}\leq n_\text{DS} - 1$, 
we have $d(S_\text{DS})=\frac{w(S_\text{DS})}{n_\text{DS}}\leq \wmax \cdot \frac{m_\text{DS}}{n_\text{DS}}< \wmax$. 
Any vertex subset (with size more than one) inducing a connected subgraph, including the output $S$, has density of at least 
\begin{align*}
\frac{\wmin}{2} > \frac{6}{19}\wmin > \frac{6}{19}\cdot \frac{\wmin}{\wmax}\cdot d(S_\text{DS}). 
%\frac{\wmin}{2} > \frac{6}{19}\wmin > \frac{6}{19}\cdot \frac{\wmin}{\wmax}\cdot d(S_\text{DS}) >\left(\frac{6}{19}\cdot \frac{\wmin}{\wmax}-\epsilon\right)d(S_\text{DS}).  
\end{align*} 
%Note that $\wmin/2$ is an obvious lower bound on the density of a single edge (and any connected subgraph). 

\textbf{Case II: $\kappa_\mathrm{max} \geq 2$.} Let us define $t = \left\lfloor \frac{12}{19}\cdot \frac{m_\text{DS}}{n_\text{DS}} \right\rfloor$. 
As $m_\text{DS}\leq {n_\text{DS}\choose 2}$ holds, we have $t < \frac{2}{5} n_\text{DS}$, and thus $n_\text{DS} > \frac{5}{2}t$. 
%and $m_\text{DS} \ge \frac{19}{12} tn_\text{DS}$ edges. 
As for the value of $m_\text{DS}$, if $t\neq 0$, $m_\text{DS} \ge \frac{19}{12} tn_\text{DS}>\frac{19}{12}t(n_\text{DS} - t)$ holds. 
%As for the value of $m_\text{DS}$, if $t=0$, then $m_\text{DS}\geq 1 > \frac{19}{12} t(n_\text{DS}-t)$ holds; otherwise $m_\text{DS} \ge \frac{19}{12} tn_\text{DS}>\frac{19}{12}t(n_\text{DS} - t)$ holds. 
Thus, by Theorem~\ref{thm:Bernshteyn}, if $t \ge 2$ holds, then the subgraph $G[S_\text{DS}]$ has a $(t+1)$-vertex-connected subgraph, which is also a subgraph of $G$. 
Hence, we have $\kappa_\text{max} \geq t + 1 \ge \frac{12}{19} \cdot \frac{m_\text{DS}}{n_\text{DS}}$. 
%where the last inequality follows from $\wmin \geq 1$.  
On the other hand, if $t < 2$ holds, then $\kappa_\text{max} \geq 2> \frac{12}{19}\cdot \frac{m_\text{DS}}{n_\text{DS}}$. 
In either case, noticing that the output $S$ is $\wmin \kappa_\text{max}$-edge-connected, 
we see that $S$ has density at least 
\begin{align*}
\frac{\wmin \cdot \kappa_\text{max}}{2} \geq \wmin\cdot \frac{6}{19}\cdot \frac{m_\text{DS}}{n_\text{DS}}
\geq \frac{6}{19}\cdot \frac{\wmin}{\wmax}\cdot \frac{w(S_\text{DS})}{n_\text{DS}}
\geq \frac{6}{19}\cdot \frac{\wmin}{\wmax}\cdot d(S_\text{DS}), 
\end{align*}
which completes the proof.  
\end{proof}

\subsection{Algorithm for Problem~\ref{prob:edge} and Analysis}
Here we present an approximation algorithm for Problem~\ref{prob:edge}, 
which is an edge-connectivity counterpart of Algorithm~\ref{alg:approx_vertex}. 
Specifically, our algorithm first computes the most highly connected subgraph in terms of edge connectivity, 
i.e., $H\in \argmax\{\lambda(H)\mid \text{$H$ is a subgraph of $G$}\}$. 
This can be done using a simple recursive algorithm mentioned by Matula~\cite{Matula78}, 
which is similar to the algorithms for computing the family of maximal $k$-edge-connected subgraphs. 
Then our algorithm simply returns the subgraph if its edge connectivity is no less than $k$ and \texttt{INFEASIBLE} otherwise. 
For reference, we describe the entire procedure in Algorithm~\ref{alg:approx_edge}. 

Matula~\cite{Matula78} stated that the time complexity of the algorithm for computing the most highly connected subgraph in terms of edge connectivity is given by $O(|V|\cdot T)$, where $T$ is the computation time required to find a minimum cut of $G$. 
If we consider one of the minimum cut algorithms by Nagamochi and Ibaraki~\cite{Nagamochi_Ibaraki_92}, Stoer and Wagner~\cite{Stoer_Wagner_97}, and Frank~\cite{Frank94}, 
the time complexity becomes $O(|V|^2(|E|+|V|\log|V|))$. 
If $G$ is simple unweighted, the time complexity reduces to $O(|E||V|\log^2|V|\log\log^2|V|)$ 
using the minimum cut algorithm by Henzinger et al.~\cite{Henzinger20}. 
Clearly, Algorithm~\ref{alg:approx_edge} has the same time complexity. 

\begin{algorithm}[t]
\caption{Approximation algorithm for Problem~\ref{prob:edge}}\label{alg:approx_edge}
%\caption{Approximation Algorithm for Problem~\ref{prob:special} with Unweighted Case}\label{alg:approximation}
\SetKwInOut{Input}{Input}
\SetKwInOut{Output}{Output}
\Input{\ $G=(V,E,w)$ and $k\in \mathbb{R}_{>0}$}
%\Input{\ $G=(V,E)$ and $k\in \mathbb{Z}_{>0}$}
\Output{\ $S\subseteq V$ or \texttt{INFEASIBLE}}
%\Output{\ Either {\sc INFEASIBLE}\ or $S\subseteq V$}
%$\ell \gets  \max\{ t \mid \exists \text{ a non-empty $t$-edge-connected subgraph of } G \}$
$H\leftarrow \argmax\{\lambda(H)\mid \text{$H$ is a subgraph of $G$}\}$\;
\If{$\lambda(H) \geq k$}{
\Return the vertex set of $H$\;
}
\Else{
\Return \texttt{INFEASIBLE}\;
}
\end{algorithm}

%For an efficient implementation, we employ a primitive that given a graph $G$ and a parameter $t$, returns all maximal $t$-edge-connected subgraphs of $G$ (or ``none'' if none
%        exists).
%      The value of $\ell$ can be computed efficiently via binary search with $O(\ell)$ calls to such a primitive. To return an $\ell$-connected subgraph of $G$ we can simply return
%      an arbitrary $\ell$-connected component of $G$, which requires only one call.
%      Using the  primitive from~\cite{Akiba+13}, which runs in time $O(m \log n)$ on graphs with $m$ edges and $n$ vertices, we obtain an overall running time of
%$O(m (\log n)^2)$ for
%Algorithm~\ref{alg:approximation}.

%Let us analyze the running time of the algorithm. 
%As mentioned in the previous section, 
%we can compute (the vertex set of) a $\tau$-edge-connected subgraph in $G$ (if exists) in $O(|V|^2(|E|+|V|\log|V|))$ time, 
%using one of the minimum cut algorithms by Nagamochi and Ibaraki~\cite{Nagamochi_Ibaraki_92} 
%and Stoer and Wagner~\cite{Stoer_Wagner_97}, 
%which means that the feasibility check and each iteration of the while-loop can be done in $O(|V|^2(|E|+|V|\log|V|))$ time. 
%Moreover, it is easy to see that the number of iterations of the while-loop of the algorithm is upper bounded by $O\left(\log\frac{\max_{v\in V}\text{deg}(v)}{\epsilon}\right)$. 
%Therefore, the overall running time is given by $O\left(|V|^2(|E|+|V|\log|V|) \log\frac{\max_{v\in V}\text{deg}(v)}{\epsilon}\right)$. 

Finally we analyze the theoretical performance guarantee of Algorithm~\ref{alg:approx_edge}. 
The following corollary is an edge-connectivity counterpart of Theorem \ref{thm:Bernshteyn}: 
\begin{corollary}\label{cor:Bernshteyn}
Let $G=(V,E,w)$ be an edge-weighted graph and let $t$ be an integer with $t\geq 2$. 
If $G$ satisfies $|V|\geq \frac{5}{2}t$ and $|E|>\frac{19}{12}t(|V|-t)$, then $G$ has a $\wmin(t+1)$-edge-connected subgraph. 
\end{corollary}

Using this corollary, we can derive the approximation ratio of Algorithm~\ref{alg:approx_edge}: 
\begin{theorem}\label{thm:approx_edge}
Algorithm~\ref{alg:approx_edge} is a polynomial-time $\left(\frac{6}{19}\cdot \frac{\wmin}{\wmax}\right)$-approximation algorithm for Problem~\ref{prob:edge}. 
\end{theorem}

The proof is similar to that of Theorem~\ref{thm:approx_vertex}, and is omitted.

\section{Open Problems}
\label{sec:concl}

There are several directions for future research. 
The most interesting one is to design a polynomial-time algorithm that has a better (bicriteria or ordinary) approximation ratio. 
We wish to remark that assuming Mader's conjecture~\cite{Mader79}, which is a stronger version of Theorem~\ref{thm:Bernshteyn}, 
we can improve the approximation ratio of Algorithms~\ref{alg:approx_vertex} and \ref{alg:approx_edge}, i.e., $\frac{6}{19}\cdot \frac{\wmin}{\wmax}$, to $\frac{1}{3}\cdot \frac{\wmin}{\wmax}$. 
However, Mader~\cite{Mader79} also conjectured that the statement is best possible, 
making it unlikely to obtain an approximation ratio better than $\frac{1}{3}\cdot \frac{\wmin}{\wmax}$ via similar analysis. 
Another interesting direction is to investigate the computational complexity of Problems~\ref{prob:vertex} and~\ref{prob:edge}.

\section*{Acknowledgments}
F.B., D.G-S., and C.T. 
 acknowledge support from Intesa Sanpaolo Innovation Center, who had no role in study design, data collection and analysis, decision to publish, or preparation of the manuscript.
A.M. was supported by Grant-in-Aid for Research Activity Start-up (No.~17H07357) and Grant-in-Aid for Early-Career Scientists (No.~19K20218). 
This work was partially done while A.M. was at RIKEN AIP, Japan, and visitied ISI Foundation, Italy. 
\bibliographystyle{abbrv}
\bibliography{dense-well-connected-subgraph}

\end{document}